\documentclass[letterpaper,superscriptaddress,twocolumn,aps]{revtex4-2}

\usepackage{color}
\usepackage{amssymb,bm}
\usepackage{amsthm}
\usepackage{amsmath}
\usepackage{dsfont}
\usepackage[dvipsnames]{xcolor}
\usepackage{bbm}
\usepackage{xifthen}
\usepackage{graphicx}
\usepackage{nicefrac}
\usepackage{multirow}
\usepackage{booktabs}
\usepackage[normalem]{ulem}
\usepackage{cancel}
\usepackage{tabularx}
\usepackage{makecell}
\usepackage[caption=false]{subfig}

\usepackage{listings}

\usepackage[T1]{fontenc}
\usepackage[utf8]{inputenc}
\usepackage[colorlinks]{hyperref}
\hypersetup{colorlinks=true, urlcolor=blue}

\usepackage{accents}

\newcommand{\cE}{\mathcal{E}}
\newcommand{\cF}{\mathcal{F}}

\newcommand{\cH}{\mathcal{H}}
\newcommand{\cI}{\mathcal{I}}

\newcommand{\cL}{\mathcal{L}}

\newcommand{\cO}{\mathcal{O}}

\newcommand{\cS}{\mathcal{S}}

\newcommand{\C}{\mathbb{C}}
\newcommand{\R}{\mathbb{R}}

\newcommand{\Id}{\mathds{1}}

\newcommand{\E}[2][]{ \ifthenelse{\isempty{#1}}
  {\mathbf{\mathbb{E}}\left[#2\right]}
  {\mathbf{\mathbb{E}}_{#1}\left[#2\right]} }

\newcommand{\tens}{\otimes}

\newcommand{\tr}[1]{\mathrm{Tr}[#1]}
\newcommand{\Ptr}[2]{\mathrm{Tr}_{#1}[#2]}
\newcommand{\ket}[1]{\lvert#1\rangle}
\newcommand{\bra}[1]{\langle#1\rvert}
\newcommand{\braket}[2]{\langle#1\rvert#2\rangle}
\newcommand{\ketbra}[2]{\lvert#1\rangle\langle#2\rvert}
\newcommand{\abs}[1]{\left\lvert#1\right\vert}
\newcommand{\ws}{\hspace{0.5em}}

\newcommand{\norm}[2]{\left\|#1\right\|_{#2}}
\makeatletter
\def\norm{\@ifnextchar[{\@with}{\@without}}
\def\@with[#1]#2{\left\|#2\right\|_{#1}}
\def\@without#1{\left\|#1\right\|}
\makeatother

\theoremstyle{plain}
\newtheorem{thm}{Theorem}

\newtheorem{lem}{Lemma}

\newtheorem{innercustomlem}{Lemma}
\newenvironment{customlem}[1]
{\renewcommand\theinnercustomlem{#1}\innercustomlem}
{\endinnercustomlem}

\newtheorem{innercustomthm}{Theorem}
\newenvironment{customthm}[1]
{\renewcommand\theinnercustomthm{#1}\innercustomthm}
{\endinnercustomthm}

\definecolor{codegreen}{rgb}{0,0.6,0}
\definecolor{codegray}{rgb}{0.5,0.5,0.5}
\definecolor{codepurple}{rgb}{0.58,0,0.82}
\definecolor{backcolour}{rgb}{0.95,0.95,0.92}

\lstdefinestyle{mystyle}{
    backgroundcolor=\color{backcolour},
    commentstyle=\color{codegreen},
    keywordstyle=\color{magenta},
    numberstyle=\tiny\color{codegray},
    stringstyle=\color{codepurple},
    basicstyle=\ttfamily\footnotesize,
    breakatwhitespace=false,
    breaklines=true,
    captionpos=b,
    keepspaces=true,
    numbers=left,
    numbersep=5pt,
    showspaces=false,
    showstringspaces=false,
    showtabs=false,
    tabsize=2
}

\AtBeginDocument{
\heavyrulewidth=.08em
\lightrulewidth=.05em
\cmidrulewidth=.03em
\belowrulesep=.65ex
\belowbottomsep=0pt
\aboverulesep=.4ex
\abovetopsep=0pt
\cmidrulesep=\doublerulesep
\cmidrulekern=.5em
\defaultaddspace=.5em}

\usepackage{MnSymbol}
\renewcommand{\arraystretch}{1.4}

\definecolor{byzantium}{rgb}{0.44, 0.16, 0.39}

\newcommand{\ubar}[1]{\underaccent{\bar}{#1}}

\renewcommand{\eqref}[1]{Eq.~(\ref{#1})} 
\newcommand{\figref}[1]{Fig.~\ref{#1}} 
\newcommand{\tabref}[1]{Table~\ref{#1}} 
\newcommand{\secref}[1]{Sec.~\ref{#1}} 
\newcommand{\appref}[1]{Appendix~\ref{#1}} 

\begin{document}
\title{Toward Reliability in the NISQ Era: Robust Interval Guarantee for Quantum Measurements on Approximate States}

\author{Maurice Weber} 
\affiliation{Department of Computer Science, ETH Zürich, Universitätstrasse 6, 8092 Zürich, Switzerland} 

\author{Abhinav Anand}
\affiliation{Chemical Physics Theory Group, Department of Chemistry, University of Toronto, Canada.}

\author{Alba Cervera-Lierta}
\affiliation{Chemical Physics Theory Group, Department of Chemistry, University of Toronto, Canada.}
\affiliation{Department of Computer Science, University of Toronto, Canada.}

\author{Jakob S. Kottmann}
\affiliation{Chemical Physics Theory Group, Department of Chemistry, University of Toronto, Canada.}
\affiliation{Department of Computer Science, University of Toronto, Canada.}

\author{Thi Ha Kyaw}
\affiliation{Chemical Physics Theory Group, Department of Chemistry, University of Toronto, Canada.}
\affiliation{Department of Computer Science, University of Toronto, Canada.}

\author{Bo Li}
\affiliation{Department of Computer Science, University of Illinois, Urbana, Illinois 61801, USA}

\author{Alán Aspuru-Guzik}
\email{aspuru@utoronto.ca}
\affiliation{Chemical Physics Theory Group, Department of Chemistry, University of Toronto, Canada.}
\affiliation{Department of Computer Science, University of Toronto, Canada.}
\affiliation{Vector Institute for Artificial Intelligence, Toronto, Canada.}
\affiliation{Canadian  Institute  for  Advanced  Research  (CIFAR)  Lebovic  Fellow,  Toronto,  Canada}

\author{Ce Zhang}
\email{ce.zhang@inf.ethz.ch}
\affiliation{Department of Computer Science, ETH Zürich, Universitätstrasse 6, 8092 Zürich, Switzerland}

\author{Zhikuan Zhao}
\email{zhikuan.zhao@inf.ethz.ch}
\affiliation{Department of Computer Science, ETH Zürich, Universitätstrasse 6, 8092 Zürich, Switzerland}

\begin{abstract}
Near-term quantum computation holds potential across multiple application domains. However, imperfect preparation and evolution of states due to algorithmic and experimental shortcomings, characteristic in the near-term implementation, would typically result in measurement outcomes deviating from the ideal setting. It is thus crucial for any near-term application to quantify and bound these output errors.
We address this need by deriving robustness intervals which are guaranteed to contain the output in the ideal setting. The first type of interval is based on formulating robustness bounds as semi-definite programs, and uses only the first moment and the fidelity to the ideal state. Furthermore, we consider higher statistical moments of the observable and generalize bounds for pure states based on the non-negativity of Gram matrices to mixed states, thus enabling their applicability in the NISQ era where noisy scenarios are prevalent. 
Finally, we demonstrate our results in the context of the variational quantum eigensolver (VQE) on noisy and noiseless simulations.
\end{abstract}

\maketitle
\section{Introduction}

Today's quantum computers are characterized by a low count of noisy qubits performing imperfect operations in a limited coherence time. In this era of quantum computing, the noisy intermediate-scale quantum (NISQ) era \cite{preskill2018quantum}, researchers and practitioners alike strive to heuristically harness limited quantum resources in order to solve classically difficult problems and also to benchmark and potentially develop new quantum subroutines. A typical pattern of these NISQ algorithms \cite{bharti2021noisy}, exemplified by the seminal variational quantum eigensolver (VQE) \cite{peruzzo2014variational} and quantum approximate optimization algorithm (QAOA) \cite{farhi2014quantum}, consists of the preparation of ansatz states with a parameterized unitary circuit followed by useful classical output being extracted by means of quantum measurements, more generally as expectation values of quantum observables through repeated measurements.

The promising potential of these NISQ algorithms spans across a wide spectrum of applications, ranging from quantum chemistry, many-body physics, and machine learning to optimization and finance~\cite{bharti2021noisy}. However, as a consequence of their heuristic nature and the prevalent imperfections in near-term implementation, NISQ algorithms in practice typically produce outputs deviating from the exact and ideal setting. This unfortunate hindrance practically arises from various sources such as circuit noise and decoherence~\cite{preskill2018quantum}, limited expressibility of ansatze~\cite{nakaji2021expressibility,sim2019expressibility}, barren plateaus during optimization in variational hybrid quantum-classical algorithms \cite{mcclean2018barren,wang2020noise,marrero2020entanglement}, measurement noise and other experimental imperfections \cite{wecker2015progress,huggins2021efficient}. 
To determine the usefulness of a given NISQ application, it is thus crucial to quantify the error on the final output in the presence of a multitude of the aforementioned sources of imperfection.

In this work, we endeavour to systematically certify the reliability of quantum algorithms by deriving robustness bounds for expectation values of observable on approximations of a target state. 
To that end, based on analytical solutions to a semidefinite program (SDP), we present lower and upper bounds to expectation values of quantum observables which depend only on the fidelity with the target state and post-processing of previously obtained measurement results.
Furthermore, we take into account higher statistical moments of the observable by generalizing the Gramian method for pure states~\cite{weinhold1968lower} to generic density operators, thus extending its application to bounding output errors resulting from noisy circuits. 
Although the focus of our investigation is on errors arising from circuit imperfection, the underlying techniques are also valid for other sources of errors such as algorithmic shortcomings.
Finally, we apply these bounds to numerically obtain robustness intervals on simulated noisy and noiseless VQE for ground state energy estimation of electronic structure Hamiltonians of several molecules.
The robustness certification protocol resulting from this work is integrated with the open source \textsc{Tequila}~\cite{kottmann2021tequila} library.

The remainder of the paper is organised as follows. In~\secref{sec:robustness_bounds}, we present our main results, namely, the bounds based on SDP and the Gramian method.
In~\secref{sec:applications}, we present our numerical simulations and explain the applicability of our bounds in the context of VQE.~\secref{sec:implementation} highlights the implementation in~\textsc{Tequila} and concluding remarks are given in ~\secref{sec:discussion}.

\section{Robustness Intervals}
\label{sec:robustness_bounds}
The goal of this work is to provide techniques to compute intervals which are guaranteed to contain the expectation value of an observable $A$ under an ideal, but unavailable, target state $\sigma$. Any such interval is referred to as a robustness interval. 
More formally, instead of having access to the state $\sigma$, we assume access to the approximate state $\rho$ and which is further assumed to have at least fidelity $1-\epsilon$ with the target state $\sigma$. 
Given these assumptions, we define a robustness interval to be an interval $\cI = [\ubar{\chi},\,\Bar{\chi}] \subseteq \R$ for which it is guaranteed that
\begin{align}
    \label{eq:robustness_interval}
    \ubar{\chi}(\epsilon,\,\rho,\,A) \leq \tr{A\sigma} &\leq \bar{\chi}(\epsilon,\,\rho,\,A)
\end{align}
and which is a function of the infidelity $\epsilon$, the observable $A$, and the state $\rho$.

\subsection{Notation.}
The Hilbert space corresponding to the quantum system of interest is denoted by $\cH \equiv \C^d$ with dimension $d = 2^n$.
We use the Dirac notation for quantum states, i.e. elements of $\cH$ are written as kets $\ket{\psi} \in \cH$ with the dual written as a bra $\bra{\psi}$. The space of linear operators acting on elements of $\cH$ is denoted by $\cL(\cH)$ and elements thereof are written in capital letters $A \in \cL(\cH)$.
The set of density operators on $\cH$ is written as $\cS(\cH)\subset \cL(\cH)$ and lower case greek letters are used to denote its elements $\sigma \in \cS(\cH)$ which are positive semidefinite and have trace equal to $1$. 
For an element $A \in \cL(\cH)$ we write $A \geq 0$ if it is positive semidefinite, $A^T$ stands for its transpose, and $A^\dagger$ is the adjoint.
We also use the Loewner partial order on the space of Hermitian operators, i.e. for two Hermitian operators $A,\,B \in \cL(\cH)$, we write $A \geq B$ if and only if $A-B \geq 0$.
Expectation values of observables, i.e. Hermitian operators $A\in\cL(\cH)$, are written as $\langle A \rangle_\sigma = \tr{A\sigma}$ for some $\sigma\in\cS(\cH)$. The variance of an observable is given by $(\Delta A_\sigma)^2 = \langle A^2\rangle_\sigma - \langle A\rangle_\sigma^2$. We write $\norm[1]{A} = \tr{\abs{A}}$ with $\abs{A} = \sqrt{A^\dagger A}$ for the trace norm of an operator $A\in\cL(\cH)$. The fidelity between quantum states $\sigma,\,\rho\in\cS(\cH)$ is defined as $\cF(\rho,\,\sigma)= \max_{\psi_\rho,\,\psi_\sigma}\abs{\braket{\psi_\rho}{\psi_\sigma}}^2$ where the maximum is taken over all purifications of $\rho$ and $\sigma$. For pure states, the fidelity reduces to the squared overlap $\cF(\ketbra{\psi}{\psi},\,\ketbra{\phi}{\phi}) = \abs{\braket{\psi}{\phi}}^2$. Finally, the real part of a complex number $z\in\C$ is written as $\Re(z)$ and the imaginary part as $\Im(z)$.

\subsection{Summary of technical results}
\begin{figure}[tbp!]
    \centering
    \includegraphics[width=\linewidth]{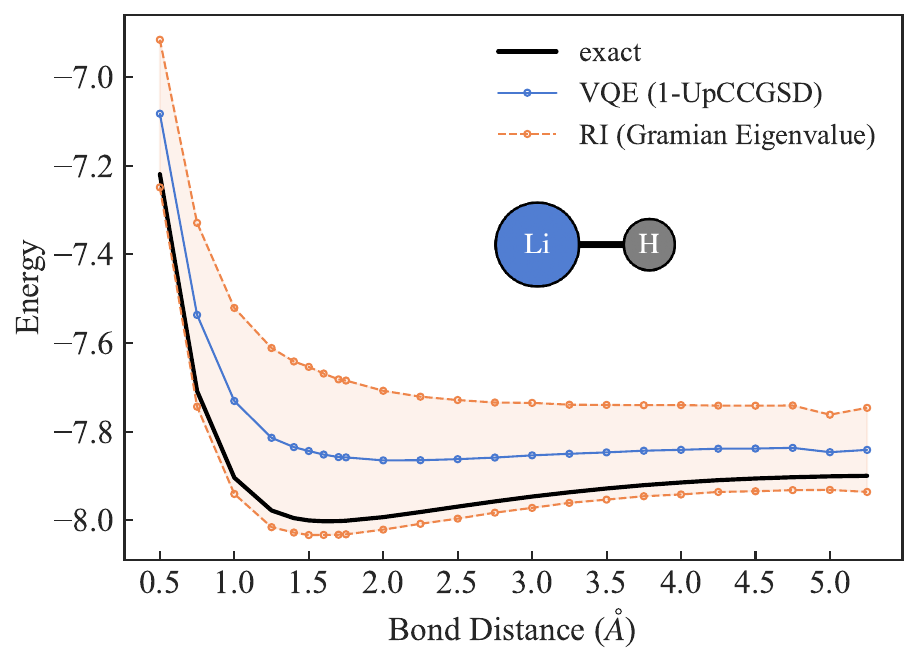}
    \caption{Bond dissociation curves and robustness interval (RI) for Lithium Hydride in a basis-set-free approach~\cite{kottmann2021reducing,kottmann2020direct}. The exact, theoretical energies are shown in black, the energy estimates provided by a noisy VQE with an UpCCGSD Ansatz~\cite{Lee2019generalized} is shown in blue. The robustness interval is guaranteed to contain the true ground state energy and is based on the Gramian eigenvalue bounds for mixed states (Theorem~\ref{thm:gramian_eigenvalue}).}
    \label{fig:lih_bounds}
\end{figure}
We employ two different techniques to bound the true expectation value $\langle A \rangle_\sigma$, each with its advantages and disadvantages in terms of efficiency and accuracy. 
The first technique is based on the formulation of lower and upper bounds as SDPs and makes use of a closed form solution of optimal type-II error probabilities from quantum hypothesis testing \cite{weber2021optimal}.
The second technique is based on the non-negativity of the determinant of Gramian matrices for a suitable collection of vectors. This second technique was initially proposed by Weinhold~\cite{weinhold1968lower} in the context of pure states. Using Uhlmann's Theorem~\cite{uhlmann1976transition}, which relates the fidelity between two mixed states to the trace norm, we extend these results to mixed states. This ultimately justifies their applicability in the current NISQ era, where the assumption of a closed quantum system is violated and one needs to make use of the density operator formalism to accurately model these states and their evolutions.

In~\tabref{table:bounds_summary}, we summarize all the results, together with the conditions under which they apply and the quantities that are covered.~\figref{fig:lih_bounds} shows the ground state energies of molecular Lithium Hydride in the basis-set-free approach of Refs.~\cite{kottmann2021reducing,kottmann2020direct}, with energy estimates provided by VQE with an UpCCGSD ansatz. The lower and upper bounds on the true energy are obtained via the Gramian method from Theorem~\ref{thm:gramian_eigenvalue}.

\begin{table*}\centering
	\begin{tabularx}{\textwidth}{@{}l X c X c X c @{}}\toprule
		&& SDP &  & \multicolumn{3}{c}{Gramian}\\
		\cmidrule(lr){3-3} \cmidrule(lr){5-7}
		 && Expectation $\langle A\rangle_\sigma$ & & Expectation $\langle A\rangle_\sigma$ && Eigenvalue $\lambda$\\
		\midrule
		Lower Bound && $(1 - 2\epsilon)\langle A\rangle_\rho - 2\sqrt{\epsilon(1 - \epsilon)(1 - \langle A \rangle_\rho^2)}$ && $(1 - 2\epsilon)\langle A\rangle_\rho -2 \sqrt{\epsilon(1-\epsilon)}\Delta A_\rho + \frac{\epsilon\langle A^2\rangle_\rho}{\langle A\rangle_\rho}$ && $\langle A \rangle_\rho - \Delta A_\rho \sqrt{\frac{\epsilon}{1-\epsilon}}$ \\
		Upper Bound && $(1 - 2\epsilon)\langle A\rangle_\rho + 2\sqrt{\epsilon(1 - \epsilon)(1 - \langle A \rangle_\rho^2)}$ && --- && $\langle A \rangle_\rho + \Delta A_\rho \sqrt{\frac{\epsilon}{1-\epsilon}}$ \\
		Assumptions && $-\Id \leq A \leq \Id$ && $A \geq 0$ && $\sigma = \ketbra{\psi}{\psi}\ \land \ A\ket{\psi} = \lambda \ket{\psi}$ \\
		\bottomrule
	\end{tabularx}
	\caption{Overview of bounds for the true expectation values and eigenvalues of a Hermitian operator $A$, with $\sigma$ the target state and $\rho$ the approximation. For the eigenvalue bound, $\sigma=\ketbra{\psi}{\psi}$ is the density operator corresponding to the eigenstate $\ket{\psi}$ with eigenvalue $\lambda = \bra{\psi}A\ket{\psi}$.
	We remark that the SDP lower and upper bounds are valid for fidelities with $\cF(\rho,\,\sigma) \geq 1 - \epsilon$ for $\epsilon \geq 0$ such that $\epsilon \leq \frac{1}{2}(1+\langle A\rangle_\rho)$ and $\epsilon \leq \frac{1}{2}(1-\langle A\rangle_\rho)$, respectively. The Gramian lower bound for expectation values is valid for $\epsilon\geq 0$ with $\sqrt{(1-\epsilon)/\epsilon} \geq {\Delta A_\rho}/{\langle A\rangle_\rho}$.}
	\label{table:bounds_summary}
\end{table*}

\subsection{Bounds via Semidefinite Programming}
Here we derive a robustness interval which is based on expressing lower and upper bounds as a semidefinite program which we connect to optimal type-II error probabilities for binary quantum hypothesis testing (QHT).

\paragraph{Quantum Hypothesis Testing.}
Binary QHT can be formulated in terms of state discrimination where two states have to be discriminated through a measurement. On a high level, the goal is to decide whether a quantum system is either in the state $\rho$, referred to as the null hypothesis, or in the state $\sigma$, the alternative hypothesis. Any such hypothesis test is represented by an operator $0 \leq \Lambda \leq \Id$ which corresponds to rejecting the null $\rho$ in favor of the alternative $\sigma$. 
The central quantities of binary QHT are the two different probabilities of making an error, namely the type-I and type-II error probabilities, defined as
\begin{align}
    \alpha(\Lambda;\,\rho) &:= \tr{\Lambda\rho}\tag{type-I error}\\
    \beta(\Lambda;\,\sigma) &:= \tr{(\Id - \Lambda)\sigma}\tag{type-II error}
\end{align}
and which quantify the probability of rejecting the null hypothesis when it is true, and accepting the null when the alternative is true, respectively.
One seeks a test $\Lambda$ which minimizes the probability of making a type-II error under the constraint that the type-I error is below some predefined threshold $\alpha_0$. Formally, we have the SDP
\begin{equation}
    \beta^*(\alpha_0;\,\rho,\,\sigma) := \sup_{0 \leq \Lambda \leq \Id}\{\beta^*(\Lambda;\,\rho,\,\sigma)\,\colon\,\alpha(\Lambda;\,\rho) = \alpha_0\}.
\end{equation}
The following Lemma establishes a closed form solution for this SDP for pure states, and a lower bound for the general case of mixed states.
\begin{lem}
\label{lem:closed_form_sdp}
	Let $\sigma,\,\rho\in\cS(\cH_d)$ be arbitrary quantum states, $\alpha_0\in[0,\,1]$ and suppose that $\cF(\rho,\,\sigma) \geq 1 - \epsilon$ for $\epsilon \leq 1 - \alpha_0$. We have
	\begin{equation}
	    \label{eq:closed_form_sdp}
	   \begin{aligned}
	       \beta^*(\alpha_0;\,\rho,\,\sigma) \, \geq \, &\alpha_0(2\epsilon - 1)+(1-\epsilon)\\ &\hspace{3em} - 2\sqrt{\alpha_0\epsilon(1-\alpha_0)(1-\epsilon)}
	   \end{aligned}
	\end{equation}
	with equality if the states $\sigma$ and $\rho$ are pure and $\cF(\rho,\,\sigma) = 1 - \epsilon$.
\end{lem}
In the following, we use this result to get closed form solutions for robustness bounds formulated as SDPs.

\subsubsection{Robustness Interval}
Consider a bounded observable $-\Id \leq A \leq \Id$ and let $\rho$ be the approximate state, corresponding to the alternative hypothesis, and let $\sigma$ be the target state, corresponding to the null hypothesis.
We can express lower and upper bounds to $\langle A\rangle_\sigma$ as semidefinite programs which take into account measurements of $\rho$. Namely, we have the upper bound
\begin{equation}
    \label{eq:sdp_upper_bound}
    \langle A\rangle_\sigma \leq \sup_{-\Id \leq \Lambda \leq \Id}\{\langle \Lambda \rangle_\sigma\,\colon\, \langle\Lambda\rangle_\rho =\langle A\rangle_\rho\}
\end{equation}
and the lower bound
\begin{equation}
    \label{eq:sdp_lower_bound}
    \langle A\rangle_\sigma \geq \inf_{-\Id \leq \Lambda \leq \Id}\{\langle \Lambda \rangle_\sigma\,\colon\,\langle\Lambda\rangle_\rho =\langle A\rangle_\rho\}.
\end{equation}
It is straight forward to see that these optimization problems are indeed valid lower and upper bounds to $\langle A\rangle_\sigma$ by noting that the operator $A$ is feasible.
In addition, as shown in Appendix~\ref{sec:appendix-tightness}, the tightness of the bounds is an immediate consequence of the formulation of the robustness interval as an SDP.
We can rewrite these SDPs and express them in terms of optimal type-II error probabilities, so that the upper bound reads
\begin{equation}
    \begin{aligned}
        \langle A\rangle_\sigma &\leq \sup_{-\Id \leq \Lambda \leq \Id}\{\langle \Lambda \rangle_\sigma\,\colon\, \langle\Lambda\rangle_\rho =\langle A\rangle_\rho\}\\
        &= 1 - 2\beta^*\left(\frac{1 + \langle A\rangle_\rho}{2};\,\rho,\,\sigma\right)
    \end{aligned}
\end{equation}
and, similarly, for the lower bound
\begin{equation}
    \begin{aligned}
        \langle A\rangle_\sigma &\geq \inf_{-\Id \leq \Lambda \leq \Id}\{\langle \Lambda \rangle_\sigma\,\colon\,\langle\Lambda\rangle_\rho =\langle A\rangle_\rho\}\\
        &= 2\beta^*\left(\frac{1 - \langle A\rangle_\rho}{2};\,\rho,\,\sigma\right) - 1.
    \end{aligned}
\end{equation}
This establishes a close connection between state discrimination via hypothesis testing, and the robustness of expectation values under perturbations to states.
Indeed, these robustness bounds formalize the intuition that states which are hard to discriminate, i.e., which admit higher error probabilities, will have expectation values which are closer together.
Furthermore, this connection also has the interesting interpretation that, if the approximate expectation $\langle A\rangle_\rho$ is close to the extreme $-1$, a statistical hypothesis test is restricted to have type-I error probability close to $0$. This makes it harder for the corresponding optimal type-II error probability $\beta^*$ to be low and hence $\langle A\rangle_\sigma$ will generally be closer to $\langle A\rangle_\rho$.
Finally, Lemma~\ref{lem:closed_form_sdp} provides a closed form solution to the SDP $\beta^*$, which only depends on the fidelity between $\rho$ and $\sigma$ and hence establishes a robustness interval of the form~\eqref{eq:robustness_interval}. This result is summarized in the following Theorem:
\begin{thm}
    \label{thm:main}
    Let $\sigma,\,\rho\in\cS(\cH_d)$ be density operators with $\cF(\rho,\,\sigma) \geq 1-\epsilon$ for some $\epsilon \geq 0$. Let $A$ be an observable with $-\Id \leq A \leq \Id$ and with expectation value $\langle A \rangle_\rho$ under $\rho$. For $\epsilon \leq \frac{1}{2}(1 + \langle A \rangle_\rho)$, the lower bound of $\langle A \rangle_\sigma$ can be expressed as
    \begin{equation}
        \langle A \rangle_\sigma \geq (1-2\epsilon)\langle A\rangle_\rho - 2\sqrt{\epsilon(1 - \epsilon)(1 - \langle A \rangle_\rho^2)}.
    \end{equation}
    Similarly, for $\epsilon \leq \frac{1}{2}(1 - \langle A \rangle_\rho)$, the upper bound of $\langle A \rangle_\sigma$ becomes
    \begin{equation}
        \langle A \rangle_\sigma \leq (1 - 2\epsilon)\langle A\rangle_\rho + 2\sqrt{\epsilon(1-\epsilon)(1 - \langle A \rangle_\rho^2)}.
    \end{equation}
\end{thm}
We remark that, although used in a different context, this technique is conceptually similar to the result presented in Theorem 1 of Ref.~\cite{weber2021optimal} for adversarial robustness of quantum classification.
Stemming from the formulation as an SDP, these bounds are tight for pure states in the sense that, for each bound, there exists an observable $A$ with expectation $\langle A \rangle_\rho$ under $\rho$ and whose expectation under $\sigma$ saturates the bound. In Appendix~\ref{sec:appendix-tightness}, we give a formal, constructive proof for this statement.
Furthermore, in practice, it is typically not feasible to measure the exact value of $\langle A\rangle_\rho$ due to finite sampling errors, measurement noise, and other experimental imperfections. For this reason, one needs to rely on confidence intervals which contain the exact value of $\langle A\rangle_\rho$ with high probability. 
This can be accounted for in the bounds from Theorem~\ref{thm:main} by noting that they are monotonic in $\langle A\rangle_\rho$,
what allows us to replace the exact value by bounds which hold with high probability.
Finally, it is worth noting that, if one has access to an estimate of the fidelity, i.e. some $\epsilon > 0$ with $\cF(\rho,\,\sigma) \geq 1 - \epsilon$, this interval can be calculated by merely postprocessing previous measurement results, and hence does not cause any computational overhead.

\subsection{Bounds via non-negativity of the Gramian}
Here we employ a different technique to derive robustness bounds, taking into account the variance of the observable as an additional piece of information. The method is based on the non-negativity of the Gramian and was pioneered by Weinhold~\cite{weinhold1968lower}. 
We first give a brief overview of the Gramian method and then present the extension to mixed states.
This extension is important as the restriction to pure states hinders the applicability of this method in practice and, in particular, in the current NISQ era, where one often has to deal with noisy states that are expressed as probabilistic ensembles of pure states in the density operator formalism.

\subsubsection{Review of the Gramian method}

Consider a Hermitian operator $A\in\cL(\cH)$, a target state $\ket{\psi}$ and an approximation of this state $\ket{\phi}$.
The Gram matrix for the vectors $\ket{\psi},\,\ket{\phi},\,A\ket{\phi}$ is given by
\begin{equation}
	G \equiv \begin{pmatrix}
		1 & \braket{\psi}{\phi} & \bra{\psi}A\ket{\phi}\\
		\braket{\phi}{\psi} & 1 & \bra{\phi}A\ket{\phi}\\
		\bra{\phi}A\ket{\psi} & \bra{\phi}A\ket{\phi} & \bra{\phi}A^2\ket{\phi}
	\end{pmatrix}
\end{equation}
where, without loss of generality, it is assumed that the overlap $\braket{\psi}{\phi}$ is real (otherwise multiply each state by a global phase).
Since Gram matrices are positive semidefinite (e.g., Theorem 7.2.10 in~\cite{horn1985matrix}), their determinants are non-negative, $\mathrm{det}(G) \geq 0$.
The function $\mathrm{det}(G)$ is a quadratic polynomial in $\Re(\bra{\psi}A\ket{\phi})$
\begin{equation}
    \begin{aligned}
        \mathrm{det}(G) &= -\Re(\bra{\psi}A\ket{\phi})^2 + 2\langle A\rangle_\phi \braket{\psi}{\phi} \Re(\bra{\psi}A\ket{\phi})\\ &\hspace{1.5em}+ (\Delta A_\phi)^2 - \abs{\braket{\phi}{\psi}}^2\langle A^2\rangle_\phi - \Im(\bra{\psi}A\ket{\phi})^2
    \end{aligned}
\end{equation}
where $(\Delta A_\phi)^2 = \bra{\phi} A^2 \ket{\phi} - \bra{\phi} A \ket{\phi}^2$ is the variance of $A$ under $\ket{\phi}$.
This polynomial vanishes when $\Re(\bra{\psi}A\ket{\phi})$ takes the values
\begin{equation}
    \braket{\phi}{\psi}\langle A\rangle_\phi \pm \sqrt{(\Delta A_\phi)^2(1 - \abs{\braket{\phi}{\psi}}^2) - \Im(\bra{\psi}A\ket{\phi})^2}
\end{equation}
and the non-negativity of $G$ thus limits the permissible values of $\Re(\bra{\psi}A\ket{\phi})$ to be within these boundaries. Since $\Im(\bra{\psi}A\ket{\phi})^2 \geq 0$ we have the bounds
\begin{equation}
    \label{eq:gramian_inequality}
    \begin{aligned}
        &\braket{\phi}{\psi}\langle A\rangle_\phi - \Delta A_\phi\sqrt{1 - \abs{\braket{\phi}{\psi}}^2}
        \leq \Re(\bra{\psi}A\ket{\phi}) \leq\\
        &\hspace{7em}\leq \braket{\phi}{\psi}\langle A\rangle_\phi + \Delta A_\phi\sqrt{1 - \abs{\braket{\phi}{\psi}}^2}
    \end{aligned}
\end{equation}
Starting from these inequalities, bounds for expectation values of quantum observables have been derived for pure states~\cite{weinhold1968lower,weinhold1968improved,weinhold1969new,weinhold1970variational,wang1971upper,blau1973upper,marmorino2002lower,marmorino2016lower} and in the context of classical methods for quantum chemistry. While certainly useful, this leaves a gap for these bounds to be applied in the NISQ era where noise is prevalent and quantum states and their evolutions are described by density operators. In the following section, we fill this gap and extend the technique to mixed states.

\subsubsection{Gramian method for mixed states}
Here, we build on the Gramian method and derive bounds which are valid for mixed states and which, in contrast to Theorem~\ref{thm:main}, take into account the second moment of the observable of interest. In principle, as more information is included, one can expect that this results in a tighter bound at the cost of having to measure additionally the expectation value of $A^2$. The following theorem provides a lower bound to expectation values of non-negative observables:
\begin{thm}[Expectation values]
    \label{thm:gramian_expectation}
    Let $\sigma,\,\rho\in\cS(\cH_d)$ be density operators with $\cF(\rho,\,\sigma) \geq 1-\epsilon$ for some $\epsilon \geq 0$. Let $A \geq 0$ be an arbitrary observable with expectation value $\langle A \rangle_\rho$ under $\rho$. For $\epsilon$ with $\sqrt{(1-\epsilon)/\epsilon} \geq \Delta A_\rho/\langle A \rangle_\rho$, the lower bound of $\langle A \rangle_\sigma$ can be expressed as
    \begin{equation}
        \label{eq:gramian_expectation_bound}
        \langle A\rangle_\sigma \geq (1 - 2\epsilon)\langle A\rangle_\rho -2 \sqrt{\epsilon(1-\epsilon)}\Delta A_\rho + \frac{\epsilon\langle A^2\rangle_\rho}{\langle A\rangle_\rho}.
    \end{equation}
\end{thm}
In the case where the target state $\sigma$ is an eigenstate of an observable $A$, the Gramian method allows to derive a further bound. While the assumptions here are stronger, this bound is particularly useful in applications such as the variational quantum eigensolver and when the observable of interest commutes with a Hamiltonian $H$ for which the target state is an eigenstate. Formally, we have the following result:
\begin{thm}[Eigenvalues]
    \label{thm:gramian_eigenvalue}
    Let $\rho\in\cS(\cH_d)$ be a density operator and let $A$ be an arbitrary observable with eigenstate $\ket{\psi}$ and eigenvalue $\lambda$, $A\ket{\psi} = \lambda\ket{\psi}$. Suppose that $\epsilon \geq 0$ is such that $\cF(\rho,\,\ket{\psi}) = \bra{\psi}\rho\ket{\psi} \geq 1-\epsilon$. Then, lower and upper bounds for $\lambda$ can be expressed as
    \begin{equation}
        \label{eq:gramian_eigenvalue_bound}
        \langle A \rangle_\rho - \Delta A_\rho \sqrt{\frac{\epsilon}{1 - \epsilon}} \ \leq \ \lambda \ \leq \ \langle A \rangle_\rho + \Delta A_\rho \sqrt{\frac{\epsilon}{1 - \epsilon}}.
    \end{equation}
\end{thm}
The proofs for both Theorems presented in this section start from~\eqref{eq:gramian_inequality} which is first extended to mixed state via purifications and an application of Uhlmann's Theorem. The second step of the proof is a rewriting of the inequality in the case where the target state is an eigenstate of $A$ (Theorem~\ref{thm:gramian_eigenvalue}), and, in the case of Theorem~\ref{thm:gramian_expectation}, an application of the Cauchy-Schwarz inequality. The proofs are given in~\appref{sec:appendix-gramian-expectation} and~\appref{sec:appendix-gramian-eigenvalue}.

\subsection{Comparison of the bounds}
We have seen three different methods to derive robustness intervals. Namely, the interval based on SDP given in Theorem~\ref{thm:main}, the expectation value lower bound based on the Gramian method from Theorem~\ref{thm:gramian_expectation}, and the robustness interval for eigenvalues from Theorem~\ref{thm:gramian_eigenvalue}, which is also based on the Gramian method.
As a first observation, we notice that the SDP bounds are dependent only on the first moment of the observable, while the bounds derived from the Gramian method take into account the second moment via the variance. In principle, this hints at a trade-off between accuracy and efficiency. That is, by taking into account higher moments, which comes at a higher computational cost, one can hope for an improvement in accuracy as more information is included. On the other hand, the SDP bounds can be calculated as a postprocessing step and thus do not require to measure additional statistics. However, as less information is included, this typically comes at the cost of lower accuracy.

On the practical side, one needs to consider that for the SDP bounds to be applicable, it is required that the observable lies between $-\Id$ and $\Id$. In practice, however, this is not always the case and the observable needs to be appropriately rescaled, e.g. by using its eigenvalues. As the exact eigenvalues might not be available, one needs to use lower and upper bounds on these, which results in a loss in tightness. This is because the set of feasible points in the SDP from~\eqref{eq:sdp_lower_bound} and \eqref{eq:sdp_upper_bound} becomes larger and hence loosens the bounds.
A similar issue emerges in the lower bound for expectation values based on the Gramian method where the observable needs to be positive semidefinite. If this assumption is violated, one again needs to apply an appropriate transformation of the observable, leading to a potentially looser bound. Instead of scaling, one can decompose the observable into individual terms, each satisfying the constraints, and then bound each term separately and aggregate the bounds over the decomposition. In~\secref{sec:applications} we consider such a decomposition in the context of VQE. Specifically, we decompose the underlying Hamiltonian into groups of mutually commuting Pauli terms and bound the expectation of each group separately.
In contrast, the eigenvalue bound based on the Gramian method does not suffer from these issues and it is applicable for general observables. It is worth remarking that this comes at the cost of less generality in the sense that the bound only applies to eigenvalues rather than general expectation values.

Assuming that the observable $A=P$ is a projection, satisfying $P^2 = P$, allows for a direct comparison between the bounds. Note that, in this case, the variance is fully determined by the first moment via $(\Delta P_\rho)^2 = \langle P\rangle_\rho - \langle P\rangle_\rho^2$ and we expect that the Gramian Expectation bound should not be tighter than the SDP bound.
First, we incorporate the knowledge that $P$ is a projection in the SDP lower bound by applying it to the observable $2P - \Id$ so that we have the bound
\begin{equation}
    \langle P\rangle_\sigma \geq (1 - 2\epsilon)\langle P\rangle_\rho + \epsilon -2 \sqrt{\epsilon(1-\epsilon)(\langle P\rangle_\rho - \langle P\rangle_\rho^2)}
\end{equation}
which is exactly the same as the lower bound derived via the Gramian method in Theorem~\ref{thm:gramian_expectation} when applied to the projection $P$.
We can also compare this bound to the Gramian eigenvalue bound from Theorem~\ref{thm:gramian_eigenvalue}. Since the latter is less general, in the sense that it only holds for target states which are eigenstates, we expect this to be tighter than the expectation value bound. As can be seen from~\figref{fig:bound_comparison_projections}, this is indeed the case.

\begin{figure}[tbp!]
    \centering
    \includegraphics[width=\linewidth]{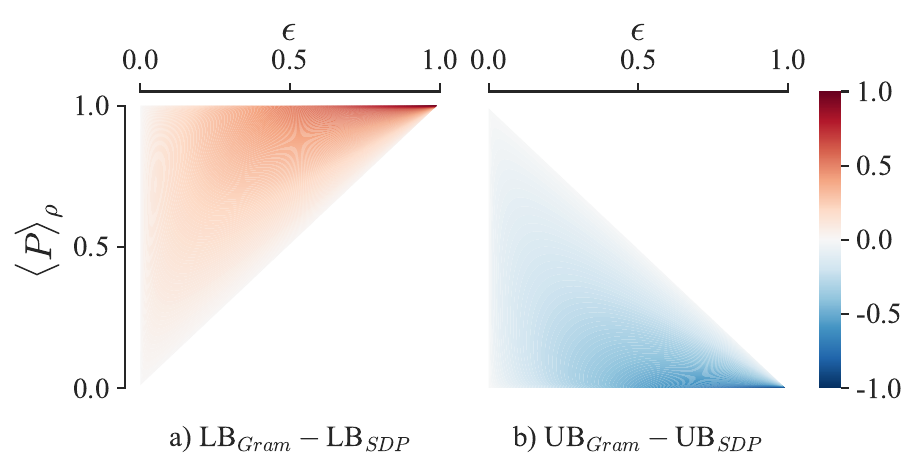}
    \caption{Difference between the SDP bounds from Theorem~\ref{thm:main} and the Gramian eigenvalue bound from Theorem~\ref{thm:gramian_eigenvalue} as a function of the infidelity $\epsilon$ and the expectation value $\langle P\rangle_\rho$. The observable is assumed to be a projection $P$ and the target state is an eigenstate of the observable. The difference is calculated by subtracting the SDP bound from the Gramian bound. a) shows the difference between lower bounds, b) shows the difference between the upper bounds. As can be seen from the figures, the Gramian eigenvalue bound is always more accurate than the expectation bound. Note that the Gramian expectation value lower bound (Theorem~\ref{thm:gramian_expectation}) equals the SDP lower bound under these assumptions.}
    \label{fig:bound_comparison_projections}
\end{figure}

Finally, we notice that all of the above bounds are faithful in the sense that, as the approximation error vanishes $\epsilon \to 0$, the bounds converge to the true expectation value $\langle A \rangle_\sigma$. To compare the rate of convergence, consider the case of pure states with the target state given by $\sigma=\ketbra{\psi}{\psi}$ and the approximation state $\rho=\ketbra{\phi}{\phi}$ with $\ket{\phi} = \sqrt{1 - \epsilon}\ket{\psi} + \sqrt{\epsilon}\ket{\psi^\perp}$ where $\ket{\psi^\perp}$ is orthogonal to $\ket{\psi}$ so that $\cF(\rho,\,\sigma) = 1 - \epsilon$. With this, one can explicitly show that the error between each bound and the true expectation $\langle A \rangle_\sigma$ scales with $\cO(\sqrt{\epsilon})$ as $\epsilon \to 0$. For values of $\epsilon$ close to $1$ on the other hand, we remark that both expectation value bounds tend towards the trivial bounds, namely $0$ for the expectation value bound, and $\pm1$ for the SDP bounds. This ultimately stems from the underlying assumptions required for the bounds to hold. In contrast, the Gramian eigenvalue bound has no assumptions on the observable $A$ and the bounds diverge as $\epsilon$ approaches $1$.

\subsection{Fidelity estimation}
\label{sec:fidelity_estimation}

All bounds presented so far have in common that they depend on the fidelity with the target state $\sigma$. However, in many practical settings, it is not possible to access the target state and thus difficult to obtain even a lower bound to the true fidelity. 
Here we seek to address this topic and present lower bounds on the true fidelity for the case where the target state is the ground state of a Hamiltonian $H$.

Let $H$ be a Hamiltonian with eigenvalues $\lambda_0 \leq \lambda_1 \leq \ldots \leq \lambda_d$ and assume that $\lambda_0$ has geometric multiplicity $1$ so that the corresponding ground state $\ket{\psi_0}$ is unique.
Let $\rho$ be a possibly mixed state approximation of $\ket{\psi_0}$. If both $\lambda_0$ and $\lambda_1$ are known, one can make use of Eckart's criterion~\cite{eckart1930} to bound the fidelity via
\begin{equation}
    \label{eq:eckart_criterion}
    \cF(\ketbra{\psi_0}{\psi_0},\,\rho) = \bra{\psi_0} \rho \ket{\psi_0} \geq \frac{\lambda_1 - \langle H\rangle_\rho}{\lambda_1 - \lambda_0}.
\end{equation}
In scenarios where knowledge of the lowest lying eigenvalues $\lambda_0$ and $\lambda_1$ is available, one can thus directly lower-bound the fidelity and use~\eqref{eq:eckart_criterion} in the computation of the robustness intervals. In scenarios where one does not have full knowledge of these eigenvalues (or, in the least, corresponding bounds), Eckart's criterion cannot be directly applied. However, we can still use the inequality if less knowledge about the spectrum of $H$ is available. If it is known that the energy estimate $\langle H\rangle_\rho$ is closer to $\lambda_0$ than to $\lambda_1$ then, as an immediate consequence of Eckart's criterion, one finds that
\begin{equation}
    \label{eq:weinstein_fidelity}
    \bra{\psi_0} \rho \ket{\psi_0} \geq \frac{1}{2}.
\end{equation}
We remark that substituting~\eqref{eq:weinstein_fidelity} into the Gramian eigenvalue bounds from Theorem~\ref{thm:gramian_eigenvalue} yields the mixed state extension of the Weinstein bounds~\cite{Weinstein1934,macdonald1934} in the non-degenerate case.
If, in addition, a lower bound $\delta$ on the spectral gap is known such that $\lambda_1 - \lambda_0 \geq \delta$, then we have the bound derived in Ref.~\cite{mcclean2016theory},
\begin{equation}
    \label{eq:mcclean_fidelity_bound}
    \bra{\psi_0} \rho \ket{\psi_0} \geq 1 - \frac{\Delta H_\rho}{\delta},
\end{equation}
which is a nontrivial lower bound whenever the variance is small enough such that $\Delta H_\rho \leq \delta$. With a similar technique, one obtains a further tightening of the bound:
\begin{equation}
    \label{eq:fidelity_bound}
    \bra{\psi_0} \rho \ket{\psi_0} \geq \frac{1}{2}\left(1 + \sqrt{1 - \left(\frac{\Delta H_\rho}{\delta/2}\right)^2}\,\right),
\end{equation}
for variances with $\Delta H_\rho \leq \delta / 2$. We note that this bound has also been reported in Refs.~\cite{weinhold1970criteria,davis1970rotation} and we provide an alternative proof in the Appendix.
In practice, the bound $\delta$ on the spectral gap can also be estimated via classical methods, as for example truncated classical configuration interaction or density-matrix renormalization group techniques.
In principle, also non-variational methods like truncated coupled-cluster (and the associated equation-of-motion or linear-response variants for the excited state energies) could be applied.
In either case, the idea is to use these classical methods to compute the ground state and the first excited state energies to get an estimate of the spectral gap which can then be used for the fidelity estimation.
The classical method which is the best to choose will generally depend on the system of interest and the available computational time.
We refer the reader to~\cite{helgaker2014molecular} for a detailed treatment over some of those methods.

The above bounds hold for Hamiltonians $H$ whose lowest eigenvalue is non-degenerate. 
In~\appref{sec:appendix-fidelity} we consider the degenerate case and show that when the approximate state $\rho$ is pure, then there always exists a state $\ket{\psi}$ which is an element of the eigenspace associated with the lowest (possibly degenerate) eigenvalue, and for which the above fidelity lower bounds hold. However, if $\rho$ is a mixed state, this cannot in general be said, as is shown in the appendix with a counterexample. In summary, when the approximate state $\rho$ is allowed to be mixed, then the fidelity bounds are applicable only when the underlying Hamiltonian has a non-degenerate ground state. If, on the other hand, $\rho$ is pure, then the bounds also hold in the degenerate case. Finally, we remark that these fidelity bounds all require varying degrees of knowledge about the ground state and Hamiltonian in question. They thus can only partially address the topic of fidelity estimation in scenarios where such knowledge is not available.

At this point we would like to point out an interesting connection to variational quantum time evolution (VarQTE). In general, VarQTE is a technique to find the ground state of a Hamiltonian $H$~\cite{mcardle2019Variational,Yuan2019theory,motta2020determining} by projecting the time evolution of the initial state to the evolution of the ansatz parameters. VarQTE typically comes with an approximation error, stemming from a limited expressibility of the Ansatz state or from noise. In Ref.~\cite{zoufal2021error}, this approximation error is quantified in terms of an upper bound on the Bures distance between the evolved state and the true ground state. Since there is a one-to-one correspondence between Bures distance and fidelity, these error bounds can be converted to a lower bound on the latter. This in turn can then be used to calculate the eigenvalue and expectation bounds presented in this work. 

\section{Applications}
\label{sec:applications}
In this section, we put into practice the theoretical results presented in the previous sections and calculate the robustness intervals for ground state energies of electronic structure Hamiltonians when the approximation of the ground state is provided by VQE. We remark that, while VQE serves as an example application, our results are not limited to ground state energies but can be used in a more general context where the goal is to calculate error bounds for expectation values. Consider a Hamiltonian $H$ with Pauli decomposition
\begin{equation}
    \label{eq:pauli_sum}
    H = \sum_{j=1}^n\omega_j P_j
\end{equation}
and let $\rho$ be an approximation to the true ground state $\ket{\psi_0}$. Given $\epsilon \geq 0$ such that $\bra{\psi_0}\rho\ket{\psi_0} \geq 1-\epsilon$, and an estimate of the variance 
$\langle H^2\rangle_{\rho} - \langle H\rangle_{\rho}^2$,
it is straightforward to evaluate the Gramian eigenvalue bounds from Theorem~\ref{thm:gramian_eigenvalue}. In contrast, for the expectation value bounds derived via SDP and the Gramian method from Theorems~\ref{thm:main} and~\ref{thm:gramian_expectation}), one needs to be more careful since the Hamiltonian $H$ might violate the underlying assumptions.
To evaluate the latter, we can account for this by adding a sufficiently large constant $c$ such that $\tilde H := H + c\Id \geq 0$ and calculate the bound for $\tilde H$, before reversing the translation in order to get the desired bound for $H$. Clearly, a valid choice for $c$ is given by $-\lambda_0$ where $\lambda_0$ is the lowest eigenvalue of $H$. However, it is not always clear which constant $c$ leads to the tightest lower bound.
Similarly, to evaluate the SDP bounds, we need to apply Theorem~\ref{thm:main} to operators which are bounded between $\pm\Id$. If the full spectrum of $H$ was known, one could normalize $H$ using these eigenvalues. 
However, in the context of VQE, the spectrum is not a priori known as this is precisely the task that VQE is designed to solve, and we need a different approach for the expectation value bounds.
The idea is to partition the terms in the Pauli decomposition from~\eqref{eq:pauli_sum} into groups so that each term corresponding to a group can be normalized. To this end, we first partition $H$ into groups of mutually qubit-wise commuting terms
\begin{equation}
    H = \sum_{k=1}^M H_k,\ws H_k = \sum_{j}\omega_j^{(k)} P_j^{(k)},\ws [P_i^{(k)},\,P_j^{(k)}] = 0.
\end{equation}
Given such a representation, the spectrum of each of the $H_k$ can be calculated classically in order to scale $H_k \to \tilde H_k$ appropriately such that the assumptions for the bounds are satisfied.
One can then compute the bounds for each of the terms in the summation and get the final bounds by aggregating the individual bounds.
We further make use of the approach presented in Refs.~\cite{yen2020measuring, verteletskyi2020} where one applies a unitary transform $U_k$ to each of the $H_k$ terms so that single-qubit measurement protocols can be used. Specifically, instead of measuring $H_k$ under the state $\rho$, one measures $A_k = U_k H_k U_k^\dagger$ under the unitarily transformed $U_k \rho U_k^\dagger$. One can then scale each $A_k$ appropriately by classically computing its eigenvalues and apply the expectation value bounds (Theorems~\ref{thm:main} and~\ref{thm:gramian_expectation}) to each term separately before aggregating.
It is also worth noting the generality of Theorem 1: although in the preceding demonstration, the matrix $A$ is generally taken to be a Pauli observable for measuring the output of a quantum circuit, the condition $-\Id \leq A \leq \Id$ is satisfied much more generally (e.g. by Fermionic operators). The application of this theorem in settings without explicit Pauli decomposition would be a fruitful ground for future research.

\subsection{Numerical simulations}
\begin{figure}[tbp!]
    \subfloat[]{\includegraphics[width=\linewidth]{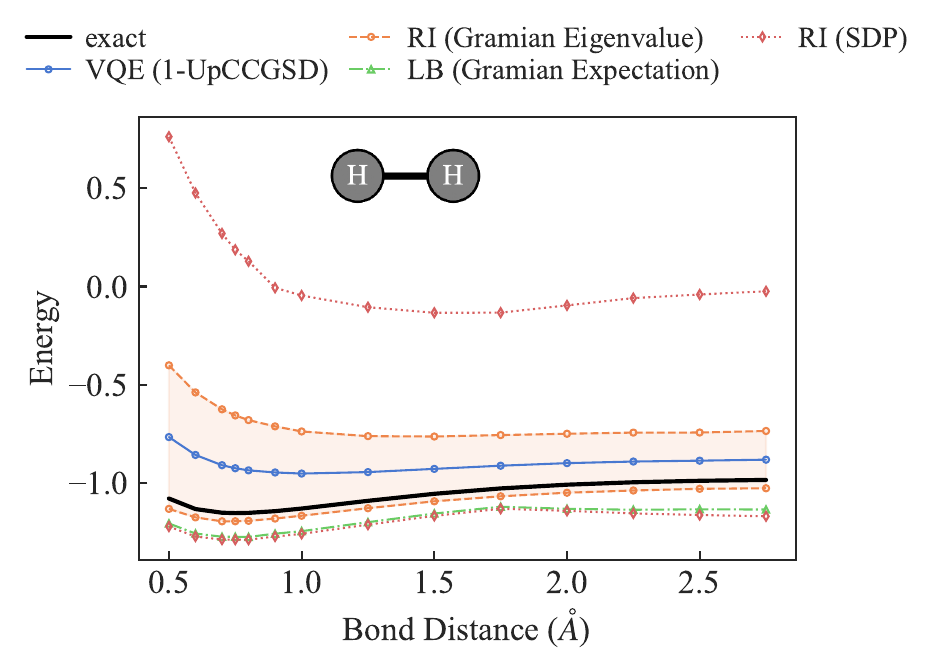}}\\
    \subfloat[]{\includegraphics[width=\linewidth]{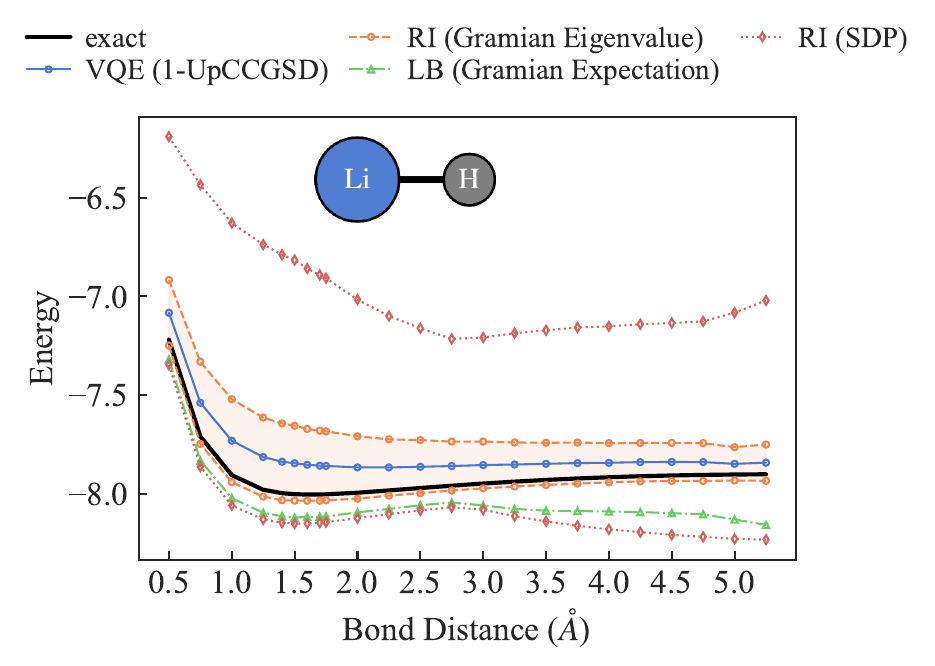}}
    \caption{Comparison of the different lower bounds (LB) and robustness intervals (RI) presented in~\secref{sec:robustness_bounds} for bond dissocation curves of H$_2$(2, 4) and LiH(2, 4). The approximation states are provided by VQE with an UpCCGSD ansatz. Both the VQE optimization and the evaluation of the bounds were simulated with bit flip and depolarization noise with $1\%$ error probability.}
    \label{fig:bounds_comparison}
\end{figure}
\begin{figure*}[tbp!]
    \centering
    \subfloat[]{\includegraphics[width=.485\linewidth]{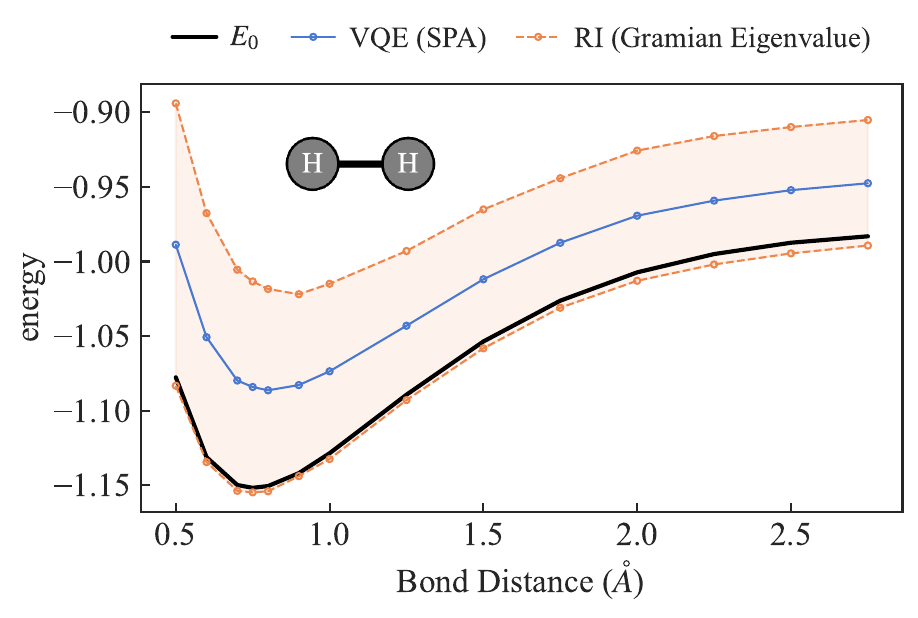}}%
    \hfill
    \subfloat[]{\includegraphics[width=.485\linewidth]{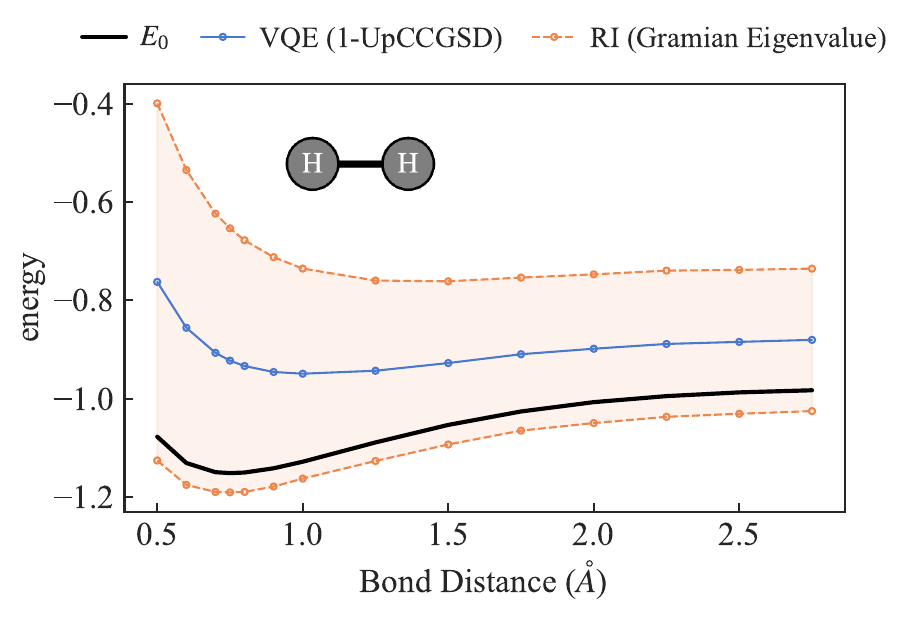}}
    
    \subfloat[]{\includegraphics[width=.485\linewidth]{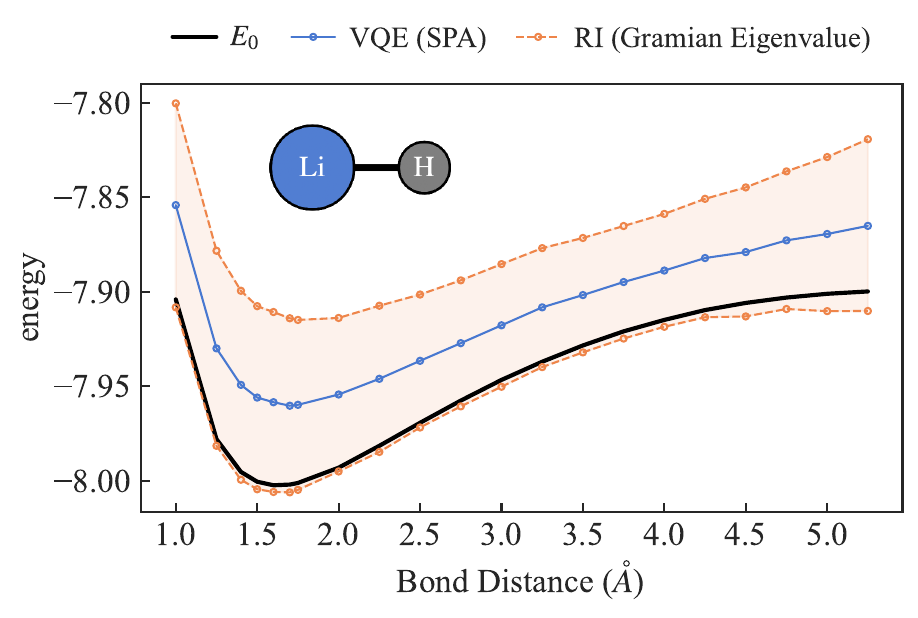}}%
    \hfill
    \subfloat[]{\includegraphics[width=.485\linewidth]{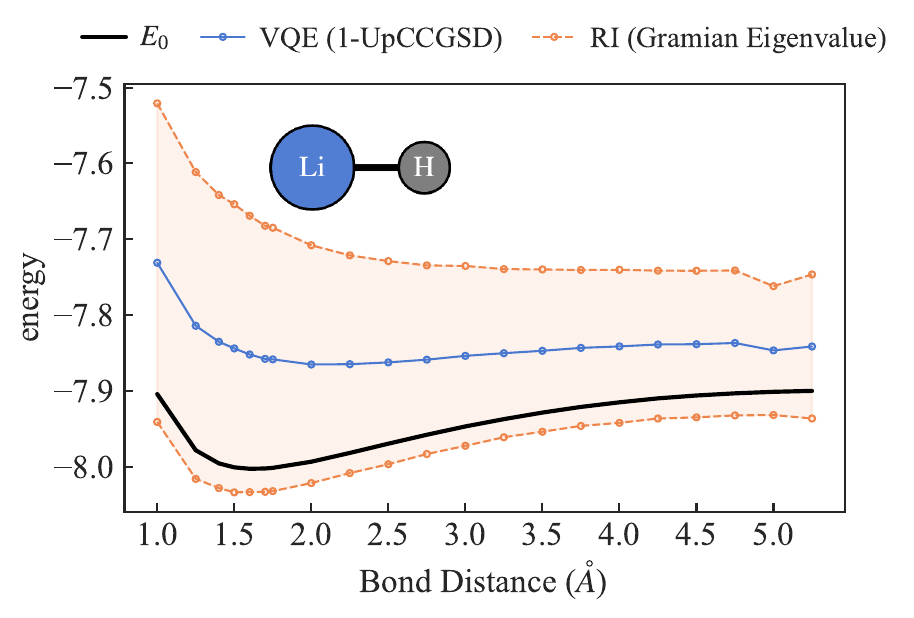}}
    
    \subfloat[]{\includegraphics[width=.485\linewidth]{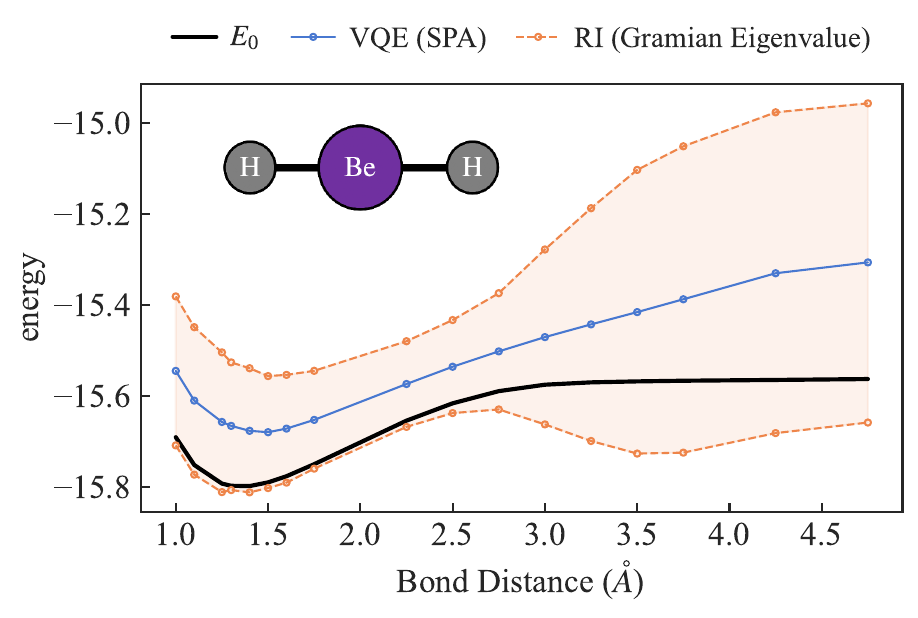}}%
    \hfill
    \subfloat[]{\includegraphics[width=.485\linewidth]{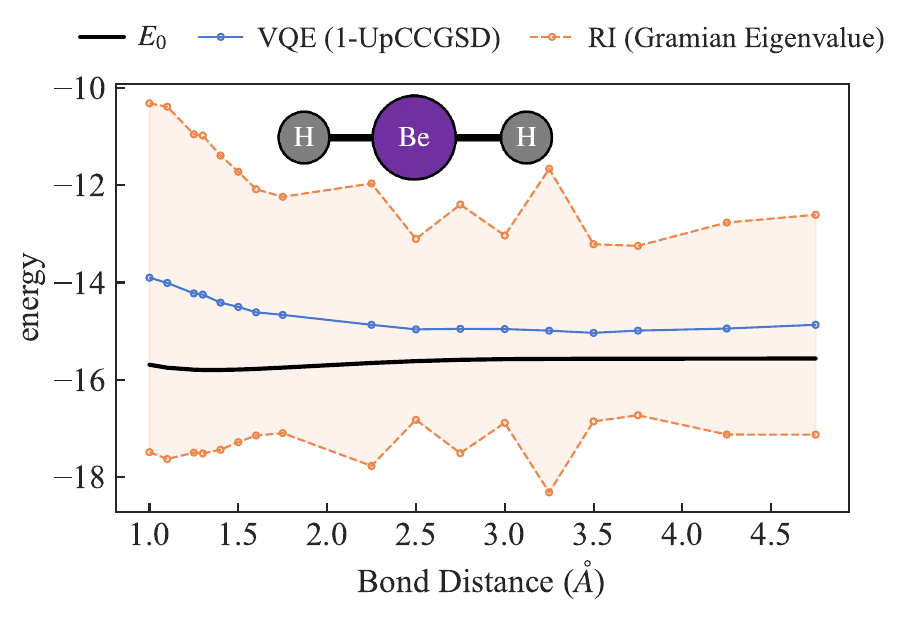}}
    
    \caption{Bond dissociation curves and robustness intervals (RI) for eigenvalues based on the Gramian method (Theorem~\ref{thm:gramian_eigenvalue}) for H$_2$(2, 4), LiH(2, 4) and BeH$_2$(4, 8).
    Both the VQE optimization and the evaluation of the bounds are done under a combination of bit flip and depolarization noise with $1\%$ error probability.}
    \label{fig:noisy_figures}
\end{figure*}

\begin{figure}[tbp!]
    \centering
    \subfloat{\includegraphics[width=\linewidth]{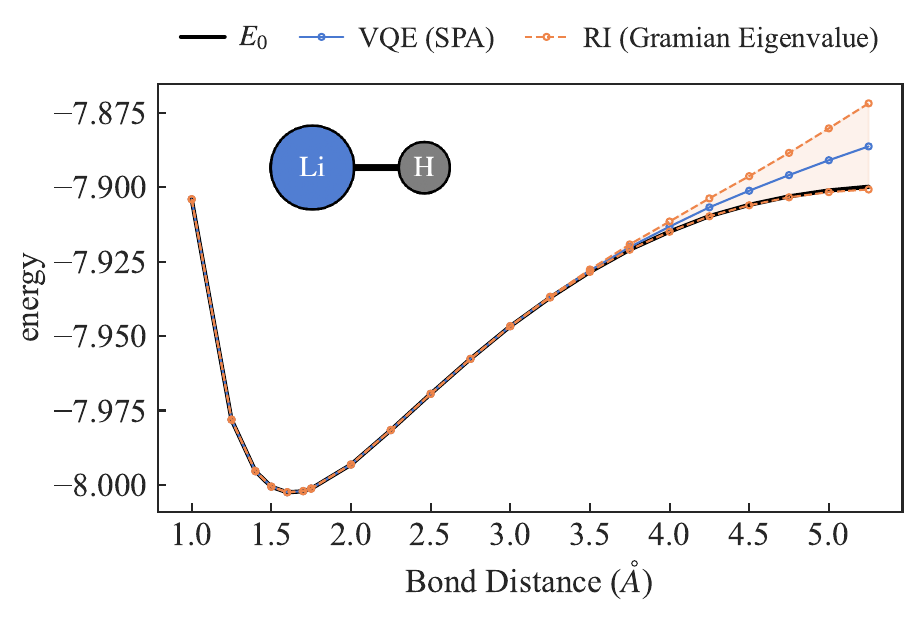}}\\
    \vspace{-1em}
    \subfloat{\includegraphics[width=\linewidth]{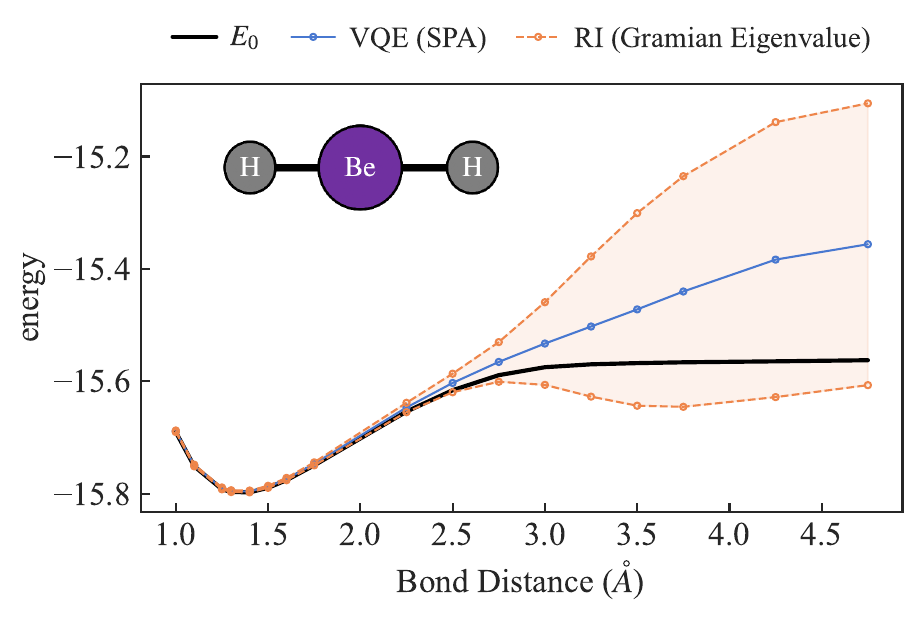}}
    \vspace{-1em}
    \caption{Bond dissociation curves, robustness interval (RI) for eigenvalues based on the Gramian method (Theorem~\ref{thm:gramian_eigenvalue}) for LiH(2, 4)and BeH$_2$(4, 8).
    Here, an ideal scenario without noise is simulated and the approximation errors stem from the limited expressibility of the Ansatz state.}
    \label{fig:noiseless_figures}
\end{figure}

Here, we numerically investigate the different robustness bounds for the ground state energies for a set of electronic structure Hamiltonians, namely H$_2$, LiH and BeH$_2$ molecules where the qubit Hamiltonians are obtained within the basis-set-free approach of Ref.~\cite{kottmann2021reducing} using directly determined pair-natural orbitals on MP2 level~\cite{kottmann2020direct}. All our experiments have been implemented in the \textsc{Tequila}~\cite{kottmann2021tequila} library using the qubit encodings from \textsc{openfermion}~\cite{openfermion}, optimizers from \textsc{scipy}~\cite{scipy}, \textsc{madness}~\cite{harrison2016madness} as chemistry backend, \textsc{qulacs}~\cite{Qulacs} as the simulation backend for noiseless simulations and \textsc{qiskit}~\cite{Qiskit} for simulations which include noise. We model noise as a combination of bitflip channels acting on single qubit gates with $1\%$ error probability, and depolarizing noise acting on two qubit gates, also with an error probability of $1\%$. 

For a given Hamiltonian $H$, we first approximate its ground state $\ket{\psi_0}$ via VQE. That is, for an Ansatz state $\rho_\theta$ with parameters $\theta$ one minimizes the objective $\langle H \rangle_{\rho_\theta}$ and obtains optimal parameters $\theta^*=\arg\min_\theta \ \langle H\rangle_{\rho_{\theta}}$. It follows from the Rayleigh-Ritz variational principle~\cite{Rayleigh1870,Ritz1909} that the expectation $\langle H \rangle_{\rho_\theta^*}$ is an upper bound to the true ground state energy $\lambda_0$. The such obtained ground state approximation $\rho_{\theta^*}$ is then used to evaluate the bounds by computing the relevant statistics, i.e. expectation values and variances of observables under this state. We notice that the quality of this state in terms of a distance to the true ground state is not easily obtainable without having some prior knowledge over the system of interest (see also~\secref{sec:fidelity_estimation} in this regard). For this reason and in order to investigate and compare the bounds, here we assume knowledge of the true fidelity with the ground state. In practice, this is of course not realistic and, as discussed previously, one needs to approximate the true fidelity. Given the ground state approximation $\rho_{\theta^*}$ and the fidelity $\cF(\rho_{\theta^*},\,\ketbra{\psi_0}{\psi_0})$, we then estimate the expectation values and variances under $\rho_{\theta^*}$ in order to evaluate the bounds. In the noiseless scenario, these statistics can be calculated exactly, whereas in the noisy scenario they need to be estimated due to finite sampling errors. Thus, in the noisy case, we repeat the calculation of the bounds 20 times and report one-sided $99\%$-confidence intervals.

In~\figref{fig:bounds_comparison}, we consider the noisy scenario and compare the different types of bounds for H$_2$(2, 4) and LiH(2, 4) with approximation states provided by VQE with an UpCCGSD Ansatz~\cite{Lee2019generalized} and optimized fermionic gradients~\cite{kottmann2021feasible}. For both molecules, we notice that the Gramian eigenvalue bound is the tightest, while the expectation value bounds are less tight. However, this is not surprising, as the eigenvalue bound is more suited for this task, compared to the other bounds which hold more generally for expectation values. In~\figref{fig:noisy_figures}, we again consider the noisy scenario and compare the Gramian eigenvalue bounds for approximation states obtained via the SPA Ansatz~\cite{kottmann2021optimized} and via the UpCCGSD Ansatz for H$_2$(2, 4), LiH(2, 4) and BeH$_2$(4, 8). We first notice that the SPA Ansatz is generally less vulnerable to noise, which stems from the associated shallow circuits, compared to the UpCCGSD Ansatz. In particular, SPA and UpCCGSD have the same expressibility for H$_2$ but, since SPA uses more efficient compiling, its energy estimates and lower bounds are more accurate compared to UpCCGSD.
In~\appref{sec:higher-error-rates} we show robustness intervals for the LiH(2, 4) molecule with the error rate increased to $10\%$. 
For the SPA ansatz, even with this error rate, the ground state fidelities vary between $0.51$ and $0.65$ while the UpCCGSD states have low ground state fidelities in the range $0.1$. In other words, UpCCGSD fails to converge to states which are close to the true ground state.
It is interesting to note that for lowest ground state fidelities, the expectation value bounds reduce to trivial bounds and the eigenvalue bound starts to diverge.
In~\figref{fig:noiseless_figures} we consider the noiseless scenario with an SPA ansatz for LiH and BeH$_2$.
In contrast to the noisy scenario, here the bounds based on the UpCCGSD Ansatz are tigther, compared to the ones based on the SPA Ansatz for large bond distances. This is due to the fact that SPA generally has more difficulties in approximating ground states for far stretched bond distances and hence results in lower ground state fidelities.
Finally, it is worth remarking that these bounds are obtained under the assumption of having complete knowledge of the true ground state fidelity, an assumption which is idealistic and typically violated in practice.

\subsection{Implementation}
\label{sec:implementation}
All our robustness intervals are implemented in the open source \textsc{Tequila}~\cite{kottmann2021tequila} library. In the following example, we run VQE for the H$_2$ Hamiltonian in a minimal representation (4 qubits), before computing the lower and upper bounds based on the optimized circuit, using the function \texttt{robustness\_interval}:

\begin{minipage}{.96\linewidth}
\begin{lstlisting}[language=Python]
import tequila as tq
from tequila.apps.robustness import robustness_interval

geom = "H .0 .0 .0\nH .0 .0 .75"
mol = tq.Molecule(geom, n_pno=1)

H = mol.make_hamiltonian()
U = mol.make_upccgsd_ansatz()
E = tq.ExpectationValue(H=H, U=U)
result = tq.minimize(E)

lower_bound, energy, upper_bound, _ = robustness_interval(U, H, fidelity, variables=result.variables)
\end{lstlisting}
\end{minipage}
Used in this way, the function calculates the robustness interval for all three methods and returns the tightest bounds. Alternatively, one can specify the type of bound via the keywords \texttt{kind} and \texttt{method} where the former stands for which kind of interval is desired, that is expectation or eigenvalue, and the latter stands for the method used to obtain the bound (Gramian or SDP). For example, calculating a robustness interval for eigenvalues using the Gramian method, can be implemented as \texttt{robustness\_interval(\ldots , kind="eigenvalue", method="gramian")}. In general, any type of expectation value can be used. Note that our implementation is agnostic with respect to the molecular representation, so that replacing \texttt{n\_pno=1} with \texttt{basis\_set="sto-3g"} will lead to a 4 qubit Hamiltonian in a traditional basis set. 

\section{Discussion and conlusions}
\label{sec:discussion}

The current experimental stage of quantum computation offers the possibility to explore the physical and chemical properties of small systems and novel quantum algorithms are being developed to extract the most from this first generation of quantum devices. 
However, this potential for computational advantage, compared to classical methods, comes at a price of noisy and imperfect simulations stemming from low qubit counts and thus the lack of quantum error correcting qubits.
The Variational Quantum Eigensolver (VQE) is the canonical example of these NISQ algorithms that allow us to obtain an approximation of Hamiltonian eigenstates by exploiting the variational principle of quantum mechanics. 
Besides the broad applications and promising results of this approach \cite{peruzzo2014variational,mcclean2016theory,arute2020hartree}, its performance guarantees should be studied and understood. 

In this work we have made first progress in this direction and have derived robustness intervals for quantum measurements of expectation values. For a target state $\sigma$, these intervals are guaranteed to contain the true expectation value $\langle A \rangle_\sigma$ of an observable $A$ when we only have access to an approximation $\rho$. Based on resource constraints, accuracy requirements, and depending on the task, we have seen three different types of robustness intervals. 
Firstly, based on the formulation of robustness bounds as SDPs, we have derived upper and lower bounds to $\langle A \rangle_\sigma$ which take into account only the first moment of the observable $A$ and can thus be obtained by post-processing measurements of $\langle A \rangle_\rho$ together with the fidelity with the target state $\cF(\rho,\,\sigma)$. Secondly, we have revisited the Gramian method~\cite{weinhold1968lower} to take into account higher statistical moments of $A$ and extended this technique to mixed states, thereby enabling their applicability in noisy scenarios which are prevalent in the NISQ era. This has led to a further lower bound to expectation values and, additionally, to lower and upper bounds on eigenvalues of observables. We have also implemented our bounds in the open source \textsc{Tequila}~\cite{kottmann2021tequila} library. Furthermore, we have validated our  results with numerical simulations of noisy and noiseless scenarios with VQE as an example application to calculate robustness intervals for ground state energies of electronic structure Hamiltonians of H$_2$, LiH and BeH$_2$. 
For the molecules considered in these experiments, we have observed that the robustness intervals provide accurate estimates of the errors incurred by noise, in particular when the ground state approximation is close enough to the true ground state in terms of fidelity.

The main requirement of the bounds obtained is the knowledge of the fidelity between the target state and its approximation. Although such a quantity is not always experimentally accessible and hence poses a challenge in the practical applicability of these bounds, there exist algorithms, such as within the variational quantum imaginary time evolution~\cite{zoufal2021error} framework, which allow for a quantification of the required approximation error in terms of distances between the target and the approximate state.
Nonetheless, our aim is to provide a formal framework to study the robustness of broadly used approaches as are the Variational Quantum Algorithms. There are still many questions around the applicability of these quantum algorithms and its robustness against noise. Within this work, we seek to unravel the uncertainties around these state-of-the-art quantum algorithms with the goal of improving its performance and applicability.

\section*{Acknowledgments}
The authors are grateful to Joseph Fitzsimons (Horizon Quantum Computing) and Nana Liu (Shanghai Jiao Tong University) for inspiring discussions on the topic of robustness of NISQ algorithms.
A.A.-G. acknowledges the generous support from
Google, Inc. in the form of a Google Focused Award.
This work is partly supported by the U.S. Department of Energy under Award No. DESC0019374 and the U.S. Office of Naval Research (ONS506661). A.A.-G. also acknowledges support from the Canada Industrial Research Chairs Program and the Canada 150 Research Chairs Program.
A.A.-G. acknowledges generous support from Anders G. Fr\"oseth and Sony Research.
CZ and the DS3Lab gratefully acknowledge the support from the Swiss National Science Foundation (Project Number 200021\_184628, and 197485), Innosuisse/SNF BRIDGE Discovery (Project Number 40B2-0\_187132), European Union Horizon 2020 Research and Innovation Programme (DAPHNE, 957407), Botnar Research Centre for Child Health, Swiss Data Science Center, Alibaba, Cisco, eBay, Google Focused Research Awards, Kuaishou Inc., Oracle Labs, Zurich Insurance, and the Department of Computer Science at ETH Zurich.

\section*{Code Availability}
The code for the numerical simulations is publicly available at~\url{https://github.com/DS3Lab/robustness-intervals-quantum-measurements}.

\bibliography{refs}

\onecolumngrid
\appendix

\renewcommand{\arraystretch}{1.0}

\section{Proof of Theorem~\ref{thm:main}}

In the following, we provide the proof for Theorem~\ref{thm:main}.
To that end, we first state the proof for Lemma~\ref{lem:closed_form_sdp} from quantum hypothesis testing and subsequently use this to obtain a closed form solution to the SDP formulation of the robustness bounds in~\eqref{eq:sdp_upper_bound} and~\eqref{eq:sdp_lower_bound}. Throughout this section, we make use of the Loewner which is defined by the convex cone of positive semidefinite matrices. That is, for Hermitian matrices $A$ and $B$, we write $A \geq B$ ($A > B$) if $A - B$ is positive semidefinite (positive definite). Furthermore, for a Hermitian operator $A$ with spectral decomposition $A = \sum_i \lambda_i P_i$, we use the notation 
\begin{equation}
\{ A \geq 0 \} := \sum_{i\colon \lambda_i \geq 0} P_i 
\hspace{1em}\text{and}\hspace{1em}
\{ A > 0 \} := \sum_{i\colon \lambda_i > 0} P_i
\end{equation}
to denote the projections onto non-negative and positive eigenspaces of $A$, respectively. The analogous notation $\{ A \leq 0 \}$, $\{ A < 0 \}$ is used for non-positive and strictly negative eigenspaces.

\subsection{Proof of Lemma~\ref{lem:closed_form_sdp}}
Here we state the proof for Lemma~\ref{lem:closed_form_sdp}, which serves as a central tool for deriving the bounds in Theorem~\ref{thm:main}. We remark that this has been shown in the appendix of Ref.~\cite{weber2021optimal}, although in a slightly different form and for pure states. For completeness, we include the proof here and refer the reader to~\cite{weber2021optimal} for proofs of specific helper lemmas. We first introduce some basic notations and definitions associated with binary quantum hypothesis testing (QHT).
For any two states $\rho,\,\sigma$ corresponding to the null and alternative hypothesis respectively, a hypothesis test is defined by an operator $0 \leq \Lambda \leq \Id$ which corresponds to rejecting the null hypothesis in favor of the alternative. The goal of binary quantum hypothesis testing is to find operators $\Lambda$ which optimally discriminate between the two states in the sense that the probability of falsely identifying $\rho$ as $\sigma$, or vice-verse, is minimized. Formally, the type-I and type-II errors are defined as
\begin{align}
    \alpha(\Lambda;\,\rho) &:= \tr{\Lambda\rho}\tag{type-I error}\\
    \beta(\Lambda;\,\sigma) &:= \tr{(\Id - \Lambda)\sigma}\tag{type-II error}
\end{align}
and can be interpreted as the probability of rejecting the null when it is true, and the probability of accepting the null when it is false. One can formalize the problem of finding an operator $\Lambda$ which minimizes the type-II error probability under the constraint the type-I error probability is below some predefined level $\alpha_0\in[0,\,1]$ as the semidefinite program
\begin{equation}
    \label{eq:type2_error_sdp}
    \beta^*(\alpha_0;\,\rho,\,\sigma) := \inf_{0 \leq \Lambda \leq \Id} \{\beta(\Lambda;\,\sigma)\,\lvert\,\alpha(\Lambda;\,\rho) \leq \alpha_0\} = \inf_{0 \leq \Lambda \leq \Id} \{\tr{(\Id - \Lambda)\sigma}\,\lvert\,\tr{\Lambda\rho} \leq \alpha_0\}.
\end{equation}
As shown in his pioneering work, Helstrom showed that this problem is solved by projections onto eigenspaces of the operator $\sigma - t\rho$~(\cite{helstrom76}, Chapter 4). These operators are called Helstrom operators and are of the form
\begin{equation}
    \Lambda_t := P_{t,+} + X_t
    \hspace{1em}\textnormal{where}\hspace{1em}
    0 \leq X_t \leq P_{t,0}
\end{equation}
and where $P_{t,+}$, $P_{t,0}$ are the projections onto the eigenspaces of $\sigma- t\rho$ associated with positive and zero eigenvalues, respectively,
\begin{equation}
    P_{t,+} := \{\sigma - t\rho > 0\},\hspace{1em}
    P_{t,-} := \{\sigma - t\rho < 0\},\hspace{1em}
    P_{t,0} := \Id - P_{t,+} - P_{t,-}.
\end{equation}
We will now review some technical Lemmas which are required for the proof of Lemma~\ref{lem:closed_form_sdp}. For proofs of these results, we refer the reader to the supplementary information of Ref.~\cite{weber2021optimal}.
\begin{lem}
    \label{lem:helstrom_optimal}
    Let $t\geq0$ and let $\Lambda_t = P_{t,+} + X_t$ for $0 \leq X_t \leq P_{t,0}$ be a Helstrom operator. Then, for any operator $0 \leq \Lambda \leq \Id$, with $\alpha(\Lambda;\,\rho) \leq \alpha(\Lambda_t;\,\rho)$, we have
    \begin{equation}
        \beta(\Lambda;\,\sigma) = \tr{(\Id - \Lambda)\sigma} \geq \tr{(\Id - \Lambda_t)\sigma} = \beta(\Lambda_t;\,\sigma)
    \end{equation}
\end{lem}
\begin{proof}
    See the proof of Lemma 6 in the supplementary information of Ref.~\cite{weber2021optimal}.
\end{proof}
The next Lemma describes some fundamental properties of $\alpha(P_{t,+};\,\rho)$ and $\alpha(P_{t,+}  + P_{t,0};\,\rho)$ as functions of $t$.
\begin{lem}
    \label{lem:type_i_properties}
    Let $P_{t,+}$ and $P_{t,0}$ be the projections onto the positive and zero eigenspaces of the operator $\sigma-t\rho$ for $t\geq 0$. Consider the functions
    \begin{equation}
        t\mapsto f(t) := \alpha(P_{t,+};\,\rho),\hspace{1em}
        t\mapsto g(t) := \alpha(P_{t,+} + P_{t,0};\,\rho).
    \end{equation}
    defined for $t\geq 0$. The functions $f$ and $g$ have the following properties:
    \begin{enumerate}
        \item[i)] $f$ is non-decreasing and continuous from the right.
        \item[ii)] $g$ is non-decreasing and continuous from the left.
    \end{enumerate}
\end{lem}
\begin{proof}
    See the proofs of Lemma 2, 3 and 4 in the supplementary information of Ref.~\cite{weber2021optimal}.
\end{proof}
Part $i)$ of this Lemma implies that the following quantity is well defined:
\begin{equation}
    t_0 := \inf\{t\geq 0 \lvert\, \tr{P_{t,+}\rho} \leq \alpha_0\}.
\end{equation}
As a consequence, we have the following sandwich inequalities:
\begin{lem}
    \label{lem:sandwich}
    For any $\alpha_0\in[0,\,1]$, we have the inequalities
    \begin{equation}
        \alpha(P_{t_0,+};\,\rho) \leq \alpha_0 \leq \alpha(P_{t_0,+} + P_{t_0,0};\,\rho).
    \end{equation}
\end{lem}
\begin{proof}
    For details, see the proof of Lemma 5 in the supplementary information of Ref.~\cite{weber2021optimal}.
\end{proof}
Now, we have all technical results to prove Lemma~\ref{lem:closed_form_sdp}.
\begin{customlem}{\ref{lem:closed_form_sdp}}[restated]
	Let $\sigma,\,\rho\in\cS(\cH_d)$ be arbitrary quantum states, $\alpha_0\in[0,\,1]$ and suppose that $\cF(\rho,\,\sigma) \geq 1 - \epsilon$ for $0 \leq \epsilon \leq 1 - \alpha_0$. We have
	\begin{equation}
       \beta^*(\alpha_0;\,\rho,\,\sigma) \, \geq \, \alpha_0(2\epsilon - 1)+(1-\epsilon) - 2\sqrt{\alpha_0\epsilon(1-\alpha_0)(1-\epsilon)}
	\end{equation}
	with equality if the states $\sigma$ and $\rho$ are pure and $\cF(\rho,\,\sigma) = 1 - \epsilon$.
\end{customlem}
\begin{proof}
    As an immediate consequence of Lemma~\ref{lem:helstrom_optimal}, for any $t\geq 0$, we have that
    \begin{equation}
        \beta^*(\alpha(\Lambda_t);\,\rho,\,\sigma) = \beta(\Lambda_t;\,\sigma).
    \end{equation}
    We now show that for any $\alpha_0 \in [0,\,1]$, there exists $t_0 \geq 0$ such that $\alpha(\Lambda_{t_0};\,\rho) = \alpha_0$.
    From part $i)$ of Lemma~\ref{lem:type_i_properties} it follows that  the quantity
    \begin{equation}
        t_0 := \inf\{t\geq 0 \lvert\, \tr{P_{t,+}\rho} \leq \alpha_0\}
    \end{equation}
    is well defined (since $\tr{P_{t,+}\rho}$ is non-increasing and continuous from the right). In addition, as a consequence of the sandwich inequalities from Lemma~\ref{lem:sandwich}, we can define the Helstrom operator
    \begin{equation}
        \Lambda_{t_0} := P_{t_0,+} + q_0\cdot P_{t_0,0},\hspace{2em}\textnormal{with}\hspace{2em}
        q_0 = \begin{cases}
            \frac{\alpha_0 - \alpha(P_{t_0,+};\,\rho)}{\alpha(P_{t_0,0};\,\rho)}, \ \ & \alpha(P_{t_0, 0};\,\rho) = 0\\
            0 & \textnormal{o.w.}\\
        \end{cases}
    \end{equation}
    which satisfies $\alpha(\Lambda_{t_0};\,\rho) = \alpha_0$. In particular, this implies that for any $\alpha_0 \in [0,\,1]$, we have
    \begin{equation}
        \beta^*(\alpha_0;\,\rho,\,\sigma) = \beta(\Lambda_{t_0};\,\sigma).
    \end{equation}
    We will now derive an explicit formula for the quantity $\beta(\Lambda_{t_0};\,\sigma)$ for pure states $\rho=\ketbra{\phi}{\phi}$ and $\sigma=\ketbra{\psi}{\psi}$, and subsequently extend it to the general case for density matrices with arbitrary rank. For the sequel, let $\gamma = \braket{\psi}{\phi}$. Since $\Lambda_{t_0}$ is a linear combination of the projections onto the eigenspaces of $\sigma - t\rho$, consider the eigenvalue problem
    \begin{equation}
        \label{eq:eigenvalue_problem}
        (\sigma - t\rho)\ket{\eta} = \eta \ket{\eta}.
    \end{equation}
    Since here $\sigma$ and $\rho$ are projections of rank one, the operator $\sigma - t\rho$ has rank at most two and there exist at most two eigenstates $\ket{\eta_0},\,\ket{\eta_1}$ corresponding to non-vanishing eigenvalues . In addition, they are linear combinations of $\ket{\psi},\,\ket{\phi}$ so that we can write
    \begin{equation}
        \ket{\eta_k} = z_{k,\psi}\ket{\psi} + z_{k,\phi} \ket{\phi},\hspace{1em} k=0,\,1
    \end{equation}
    with constants $z_{k,\psi},\,z_{k,\phi}$. Thus, solving the eigenvalue problem in~\eqref{eq:eigenvalue_problem} amounts to solving the problem
    \begin{equation}
        \begin{pmatrix}
            1 & \gamma \\
            -t\Bar{\gamma} & -t
        \end{pmatrix}
        \cdot
        \begin{pmatrix}
            z_{k,\psi} \\
            z_{k,\phi}
        \end{pmatrix}
        = \eta_k
        \begin{pmatrix}
            z_{k,\psi} \\
            z_{k,\phi}
        \end{pmatrix},
    \end{equation}
    for which we find eigenvalues
    \begin{equation}
        \begin{gathered}
            \eta_0 = \frac{1}{2}(1-t) - R_t,\hspace{1em} \eta_1 = \frac{1}{2}(1-t) + R_t
            \hspace{1em}\textnormal{with}\hspace{1em}
            R_t = \sqrt{\frac{1}{4}(1 - t)^2 + t(1 - \abs{\gamma}^2)}.
        \end{gathered}
    \end{equation}
    The corresponding eigenvectors $\ket{\eta_0},\,\ket{\eta_1}$ are determined by their coefficients $z_{k,\psi},\,z_{k,\phi}$ for $k=0,\,1$, for which we find
    \begin{equation}
        \begin{gathered}
            z_{k,\psi} = -\gamma A_k,\hspace{1em} z_{k,\phi} = (1-\eta_k)A_k,\hspace{1em} \abs{A_k}^{-2} = 2R_t\abs{\eta_k - 1 + \abs{\gamma}^2}
        \end{gathered}
    \end{equation}
    where the coefficient $A_k$ arises from requiring the eigenvectors $\ket{\eta_k}$ to be normalized (see Ref.~\cite{helstrom68}, Section 8).
    Note that $\forall t\geq 0$ we have $\eta_0 \leq 0$ and $\eta_1 > 0$ so that $P_{t,+} = \ketbra{\eta_1}{\eta_1}$. Recall that we defined $t_0$ to be the positive number $t_0:=\inf\{t\geq 0\colon\, \tr{P_{t,+}\rho} \leq \alpha_0\}$. It follows that
    \begin{equation}
        \tr{P_{t,+}\rho} = \abs{\braket{\phi}{\eta_1}}^2= \frac{1}{2}\left(1 - \frac{1 + t - 2 \abs{\gamma}^2}{\sqrt{(1 + t)^2 - 4t\abs{\gamma}^2}}\right)
    \end{equation}
    and notice that the right hand side is continuous in $t$ over $[0,\,\infty)$ whenever $\abs{\gamma} < 1$.
    Since $t\mapsto \tr{P_{t,+}\rho}$ is non-increasing, it's maximum is attained at $t=0$ so that $\tr{P_{t,+}\rho} \leq \abs{\gamma}^2$ and hence $t_0 = 0$ if $\alpha_0 > \abs{\gamma}^2$.
    In this case, we obtain $\beta^*(\alpha_0;\rho,\,\sigma) = 0$.
    If, on the other hand, $\alpha_0 \leq \abs{\gamma}^2$, then we solve the equation $\tr{P_{t,+}\rho} = \alpha_0$ and obtain the expression for $t_0$
    \begin{equation}
        t_0 = 2\abs{\gamma}^2 - 1 - (2\alpha_0 - 1)\sqrt{\frac{\abs{\gamma}^2(1 - \abs{\gamma}^2)}{\alpha_0(1-\alpha_0)}}.
    \end{equation}
    For $t=t_0$ we have $\eta_0 < 0$ and $\eta_1 > 0$ so that $\Lambda_{t_0} = \ketbra{\eta_1}{\eta_1}$ and $\ket{\eta_1} = -\gamma A_1 \ket{\psi} + (1-\eta_1)A_1\ket{\phi}$. Hence
    \begin{equation}
            \beta(\Lambda_{t_0};\,\sigma)=\, 1 - \abs{\braket{\eta_1}{\psi}}^2 = 1 - \abs{A_1}^2\abs{\gamma}^2\eta_1^2.
    \end{equation}
    Plugging $t_0$ into the expressions above yields
    \begin{equation}
        \label{eq:sdp_closed_pure_states}
        \beta^*(\alpha_0;\,\phi,\,\psi) = \beta(\Lambda_{t_0};\,\sigma)=\, \alpha_0\cdot(1 - 2\abs{\gamma}^2) + \abs{\gamma}^2 - 2 \sqrt{(1-\alpha_0)(1-\abs{\gamma}^2)\abs{\gamma}^2\alpha_0}.
    \end{equation}
    Since the right hand side of~\eqref{eq:sdp_closed_pure_states} is monotonically decreasing in $\abs{\gamma}^2$ and $\abs{\gamma}^2 \geq 1-\epsilon$, the claim follows for pure states. To see that the above expression also constitutes a valid lower bound for mixed states, let $\Psi$ and $\Phi$ be arbitrary purifications of $\sigma$ and $\rho$ respectively, both with purifying system $\cH_E$. It is well known that $\beta^*$ is monotonically increasing under the action of any completely positive and trace preserving map $\cE$, i.e. $\beta^*(\alpha_0;\sigma,\,\rho) \leq \beta^*(\alpha_0;\cE[\sigma],\,\cE[\rho])$ for any $\alpha_0 \in [0,\,1]$. Since the partial trace $\Ptr{E}{\cdot}$ is itself a CPTP map, we have the inequality
    \begin{equation}
        \label{eq:proof_lemma_purifications}
        \begin{aligned}
            \beta^*(\alpha_0;\,\rho,\,\sigma) &= \beta^*(\alpha_0;\,\Ptr{E}{\ketbra{\Phi}{\Phi}},\,\Ptr{E}{\ketbra{\Psi}{\Psi}})\\
            &\geq \beta^*(\alpha_0;\,\Phi,\,\Psi)\\
            &= \alpha_0\cdot(1 - 2\abs{\braket{\Psi}{\Phi}}^2) + \abs{\braket{\Psi}{\Phi}}^2 - 2 \sqrt{(1-\alpha_0)(1-\abs{\braket{\Psi}{\Phi}}^2)\abs{\braket{\Psi}{\Phi}}^2\alpha_0}
        \end{aligned}
    \end{equation}
    where $\Ptr{E}{\cdot}$ denotes the partial trace over the purifying system. It follows from Uhlmann's Theorem that we can choose $\Psi,\,\Phi$ such that $\abs{\braket{\Psi}{\Phi}}^2 = \cF(\rho,\,\sigma)$. The claim now follows from the observation that the RHS of~\eqref{eq:proof_lemma_purifications} is monotonically decreasing in $\abs{\braket{\Psi}{\Phi}}^2$. This completes the proof.
\end{proof}

\subsection{Main Proof}
Given Lemma~\ref{lem:closed_form_sdp}, here we provide the full proof for Theorem~\ref{thm:main}. We first restate the result.
\begin{customthm}{\ref{thm:main}}[restated]
    Let $\sigma,\,\rho\in\cS(\cH_d)$ be density operators with $\cF(\rho,\,\sigma) \geq 1-\epsilon$ for some $\epsilon \geq 0$. Let $A$ be an observable with $-\Id \leq A \leq \Id$ and with expectation value $\langle A \rangle_\rho$ under $\rho$. For $\epsilon \leq \frac{1}{2}(1 + \langle A \rangle_\rho)$, the lower bound of $\langle A \rangle_\sigma$ can be expressed as
    \begin{equation}
        \langle A \rangle_\sigma \geq (1-2\epsilon)\langle A\rangle_\rho - 2\sqrt{\epsilon(1 - \epsilon)(1 - \langle A \rangle_\rho^2)}.
    \end{equation}
    Similarly, for $\epsilon \leq \frac{1}{2}(1 - \langle A \rangle_\rho)$, the upper bound of $\langle A \rangle_\sigma$ becomes
    \begin{equation}
        \langle A \rangle_\sigma \leq (1 - 2\epsilon)\langle A\rangle_\rho + 2\sqrt{\epsilon(1-\epsilon)(1 - \langle A \rangle_\rho^2)}.
    \end{equation}
\end{customthm}
\begin{proof}
    The idea is to first formulate the robustness bounds as semidefintie programs which takes into account the first moment of $A$ under $\rho$ and the assumption that $-\Id \leq A \leq \Id$. This is then connected to the SDP~\eqref{eq:type2_error_sdp} from QHT for which a closed form lower bound is known from Lemma~\ref{lem:closed_form_sdp}. We start with the upper bound. Consider the SDP
    \begin{equation}
        \max \left\{\tr{\Lambda\sigma}\,\lvert \, \tr{\Lambda\rho} = \tr{A\rho},\, -\Id  \leq \Lambda \leq \Id \right\}
    \end{equation}
    which is an upper bound to $\langle A \rangle_\sigma$ since the operator $A$ is feasible.
    We can rewrite this as
    \begin{equation}
        \begin{aligned}
            \max \left\{\tr{\Lambda\sigma}\,\lvert \, \tr{\Lambda\rho} = \tr{A\rho},\, -\Id  \leq \Lambda \leq \Id \right\} &= -1 + 2\max \left\{\tr{\Tilde\Lambda\sigma}\,\lvert\, \tr{\Tilde\Lambda\rho} = \frac{1}{2}(1 + \tr{A\rho}),\, 0 \leq \Tilde{\Lambda} \leq \Id\right\}\\
            &= 1 - 2\beta^*\left(\frac{1}{2}(1 + \tr{A\rho});\,\rho,\,\sigma\right)
    \end{aligned}
    \end{equation}
    where the second equality follows from the fact that replacing the equality with an inequality in the constraint of the SDP~\eqref{eq:type2_error_sdp} leads to the same solution.
    It then follows from Lemma~\ref{lem:closed_form_sdp} that
    \begin{align}
        \langle A \rangle_\sigma &\leq (1-2\epsilon)\langle A\rangle_\rho + 2\sqrt{\epsilon(1 - \epsilon)(1 - \langle A \rangle_\rho^2)}
    \end{align}
    for $\epsilon \geq 0$ with $1 - \epsilon \geq \frac{1}{2}(1 + \langle A \rangle_\rho)$.
    
    To show the lower bound, consider the SDP
    \begin{equation}
        \min \left\{\tr{\Lambda\sigma}\,\lvert \, \tr{\Lambda\rho} = \tr{A\rho},\, -\Id  \leq \Lambda \leq \Id \right\}
    \end{equation}
    which is a lower bound to $\langle A \rangle_\sigma$ since the operator $A$ is feasible. We rewrite this as
    \begin{equation}
        \begin{aligned}
            \min \left\{\tr{\Lambda\sigma}\,\lvert \, \tr{\Lambda\rho} = \tr{A\rho},\, -\Id  \leq \Lambda \leq \Id \right\} &= 2\min\left\{\tr{(\Id - \Tilde\Lambda)\sigma}\, \lvert \, \tr{\Tilde\Lambda\rho} = \frac{1}{2}(1 - \tr{A\rho}),\, 0 \leq \Tilde{\Lambda} \leq \Id\right\} - 1\\
            &= 2\beta^*\left(\frac{1}{2}(1 - \tr{A\rho});\,\rho,\,\sigma\right) - 1
        \end{aligned}
    \end{equation}
    and again use Lemma~\ref{lem:closed_form_sdp} and obtain the lower bound
    \begin{align}
        \langle A \rangle_\sigma \geq (1 - 2\epsilon)\langle A\rangle_\rho - 2\sqrt{\epsilon(1 - \epsilon)(1 - \langle A \rangle_\rho^2)}
    \end{align}
    for $\epsilon \geq 0$ with $1-\epsilon \geq \frac{1}{2}(1 - \langle A \rangle_\rho)$.
\end{proof}

\subsection{Tightness of the SDP bound}
\label{sec:appendix-tightness}
Here we show that the SDP bounds presented in Theorem~\ref{thm:main} are tight for pure states in the sense that, under the given assumptions, there exists an observable which saturates the bound in the case that $\cF(\rho,\,\sigma) = 1-\epsilon$. Let us first consider the upper bound and recall that in the proof of Theorem~\ref{thm:main} we have shown
\begin{equation}
    \begin{aligned}
    \langle A \rangle_\sigma &\leq \max \left\{\tr{\Lambda\sigma}\,\lvert \, \tr{\Lambda\rho} = \tr{A\rho},\, -\Id  \leq \Lambda \leq \Id \right\}\\
    &= 1 - 2\beta^*\left(\frac{1}{2}(1 + \tr{A\rho});\,\rho,\,\sigma\right)\\
    &\leq (1-2\epsilon)\langle A\rangle_\rho + 2\sqrt{\epsilon(1-\epsilon)(1 - \langle A \rangle_\rho^2)}
    \end{aligned}
\end{equation}
where the last inequality follows from Lemma~\ref{lem:closed_form_sdp}.
Additionally, it follows from Lemma~\ref{lem:closed_form_sdp} that for pure states this inequality is indeed an equality, that is, we have shown that for pure states we have
\begin{equation}
    \beta^*(\alpha_0;\,\rho,\,\sigma) = \alpha_0(2\epsilon - 1)+(1-\epsilon) - 2\sqrt{\epsilon\alpha_0(1-\epsilon)(1-\alpha_0)}
\end{equation}
for $1-\epsilon \geq \alpha_0$ for arbitrary $\alpha_0 \in [0,\,1]$.
This followed from constructing an operator $\Lambda^\star$ with $\tr{\sigma(\Id - \Lambda^\star)} = \beta^*(\alpha_0;\,\rho,\,\sigma)$. 
Let $\Lambda^\star$ be such an operator for $\alpha_0 = \tr{\Lambda^\star\rho}$ and let $A^\star := 2\Lambda^\star - \Id$. We notice that
\begin{equation}
    \begin{aligned}
        \langle A^\star\rangle_\sigma &= \tr{\sigma(2\Lambda^\star - \Id)}\\
        & = 1 - 2\tr{\sigma(\Id - \Lambda^\star)}\\
        & = 1 - 2\beta^*\left(\tr{\Lambda^\star\rho});\,\rho,\,\sigma\right)
    \end{aligned}
\end{equation}
which shows that the bound is indeed saturated for the observable $A^\star$. Tightness of the lower bound can be shown analogously.

\section{Proof of Theorem~\ref{thm:gramian_expectation}}
\label{sec:appendix-gramian-expectation}
\begin{customthm}{\ref{thm:gramian_expectation}}[restated]
    Let $\sigma,\,\rho\in\cS(\cH_d)$ be density operators with $\cF(\rho,\,\sigma) \geq 1-\epsilon$ for some $\epsilon \geq 0$. Let $A \geq 0$ be an arbitrary observable with expectation value $\langle A \rangle_\rho$ under $\rho$. For $\epsilon$ with $\sqrt{(1-\epsilon)/\epsilon} \geq \Delta A_\rho/\langle A \rangle_\rho$, the lower bound of $\langle A \rangle_\sigma$ can be expressed as
    \begin{equation}
        \langle A\rangle_\sigma \geq (1 - 2\epsilon)\langle A\rangle_\rho -2 \sqrt{\epsilon(1-\epsilon)}\Delta A_\rho + \frac{\epsilon\langle A^2\rangle_\rho}{\langle A\rangle_\rho}.
    \end{equation}
\end{customthm}
\begin{proof}
Recall that the Gramian inequalities from~\eqref{eq:gramian_inequality} for pure states read
\begin{equation}
    \braket{\phi}{\psi}\langle A\rangle_\phi - \Delta A_\phi\sqrt{1 - \abs{\braket{\phi}{\psi}}^2}
    \ \leq \  \Re(\bra{\psi}A\ket{\phi})
    \ \leq \ \braket{\phi}{\psi}\langle A\rangle_\phi + \Delta A_\phi\sqrt{1 - \abs{\braket{\phi}{\psi}}^2}
\end{equation}
where $\psi$ denotes the target state and $\phi$ the approximate state. The first step is to show that this also holds for mixed states. Uhlmann's Theorem~\cite{uhlmann1976transition} states that for any two mixed states $\rho$ and $\sigma$, we have
\begin{equation}
    \cF(\rho,\,\sigma) = \norm[1]{\sqrt{\rho}\sqrt{\sigma}}^2.
\end{equation}
The trace norm in its variational form is given by $\norm[1]{S} = \max_{U} \abs{\tr{US}}$ for arbitrary $S \in \cL(\cH_d)$ and where the maximization is taken over all unitaries. It follows that there exists $U$ such that
\begin{equation}
    \label{eq:fidelity_rewrite}
    \cF(\rho,\,\sigma) = \abs{\tr{U\sqrt{\rho}\sqrt{\sigma}}}^2
\end{equation}
Let $\ket{\Omega} = \sum_{k=1}^d\ket{k} \tens \ket{k}$ be the unnormalized maximally entangled state on $\cH_d \tens \cH_d$ and note that the trace in the RHS of~\eqref{eq:fidelity_rewrite} can be rewritten as
\begin{equation}
    \begin{aligned}
        \tr{U\sqrt{\rho}\sqrt{\sigma}} &= \sum_k \bra{k} U\sqrt{\rho}\sqrt{\sigma} \ket{k}\\
        &=\sum_{k,l} \bra{k} U\sqrt{\rho}\sqrt{\sigma} \ket{l} \braket{k}{l}\\
        &=\sum_{k,l} \bra{k} \tens \bra{k}( U\sqrt{\rho} \tens \Id )(\sqrt{\sigma} \tens \Id)\ket{l} \tens \ket{l}\\
        &= \bra{\Omega}(\sqrt{\rho} \tens (U^T)^\dagger)(\sqrt{\sigma} \tens \Id) \ket{\Omega}
    \end{aligned}
\end{equation}
where the last equality follows from the definition of $\ket{\Omega}$ and the fact that $(\Id \tens U)\ket{\Omega} = (U^T \tens \Id)\ket{\Omega}$. Define the (pure) states
\begin{equation}
    \label{eq:purifications}
    \ket{\Psi} \equiv (\sqrt{\sigma} \tens \Id) \ket{\Omega},\hspace{2em} \ket{\Phi} \equiv (\sqrt{\rho} \tens U^T) \ket{\Omega}
\end{equation}
and note that these states are purifications of $\sigma$ and $\rho$ respectively and $\cF(\rho,\,\sigma) = \abs{\braket{\Psi}{\Phi}}^2$. Furthermore, we have
\begin{equation}
    \begin{aligned}
        \langle A \tens \Id \rangle_\Psi &= \bra{\Omega} (\sqrt{\sigma} \tens \Id)(A \tens \Id) (\sqrt{\sigma} \tens \Id)\ket{\Omega}\\
        &= \tr{\sqrt{\sigma}A\sqrt{\sigma}}\\
        &= \langle A\rangle_\sigma
    \end{aligned}
\end{equation}
and similarly
\begin{equation}
    \begin{aligned}
        \langle A \tens \Id \rangle_\Phi &= \bra{\Omega} (\sqrt{\rho} \tens (U^T)^\dagger)(A \tens \Id) (\sqrt{\rho} \tens U^T)\ket{\Omega}\\
        &=\bra{\Omega} (\sqrt{\rho}A\sqrt{\rho} \tens (U U^\dagger)^T)\ket{\Omega}\\
        &= \tr{\sqrt{\rho}A\sqrt{\rho}}\\
        &= \langle A\rangle_\rho
    \end{aligned}
\end{equation}
Replacing $A$ with $A^2$ in the above, we find
\begin{equation}
    \begin{aligned}
        (\Delta (A\tens\Id)_\Phi)^2 &= \langle A^2 \tens \Id \rangle_\Phi - \langle A \tens \Id \rangle_\Phi^2\\
        &= \langle A^2\rangle_\rho - \langle A \rangle_\rho^2\\
        &= (\Delta A_\rho)^2.
    \end{aligned}
\end{equation}
Finally, we have
\begin{equation}
    \begin{aligned}
        \bra{\Psi}(A\tens \Id)\ket{\Phi} &= \bra{\Omega}(\sqrt{\sigma} A \sqrt{\rho}U \tens \Id)\ket{\Omega}\\
        &= \tr{\sqrt{\sigma} A \sqrt{\rho}U}\\
        &= \langle A\sqrt{\sigma},\,  \sqrt{\rho}U\rangle_{HS}
    \end{aligned}
\end{equation}
where $\langle\cdot,\,\cdot\rangle_{HS}$ denotes the Hilbert-Schmidt inner product.
Without loss of generality, we assume that $\braket{\Phi}{\Psi}$ is real and non-negative since otherwise we can multiply each purificiation by a global phase. Applying the Gramian inequalities to the observable $A\tens \Id$ and the purifications $\ket{\Psi},\ \ket{\Phi}$, we find
\begin{equation}
    \label{eq:gramian_inequality_mixed}
    \begin{aligned}
        \sqrt{\cF(\rho,\,\sigma)}\langle A\rangle_\rho - \Delta A_\rho\sqrt{1 - \cF(\rho,\,\sigma)} &= \braket{\Phi}{\Psi}\langle A\tens \Id\rangle_\Phi - \Delta (A \tens \Id)_\Phi\sqrt{1 - \abs{\braket{\Phi}{\Psi}}^2}\\
        &\leq  \Re(\bra{\Psi}(A\tens \Id)\ket{\Phi})\\
        &=\Re(\langle A\sqrt{\sigma},\,  \sqrt{\rho}U\rangle_{HS})\\
        &\leq \braket{\Phi}{\Psi}\langle A\tens \Id\rangle_\Phi + \Delta (A \tens \Id)_\Phi\sqrt{1 - \abs{\braket{\Phi}{\Psi}}^2}\\
        &=\sqrt{\cF(\rho,\,\sigma)}\langle A\rangle_\rho + \Delta A_\rho\sqrt{1 - \cF(\rho,\,\sigma)}.
    \end{aligned}
\end{equation}
Thus, we have shown that inequalities similar to~\eqref{eq:gramian_inequality} also hold for mixed states. To finish the proof, note that by assumption $A \geq 0$ and hence $A$ has a square root $A = A^{1/2}A^{1/2}$. The Cauchy-Schwarz inequality yields
\begin{equation}
    \begin{aligned}
        \Re(\langle A\sqrt{\sigma},\,  \sqrt{\rho}U\rangle_{HS}) &\leq \abs{\langle A\sqrt{\sigma},\,  \sqrt{\rho}U\rangle_{HS}}\\
        &=\abs{\langle A^{1/2}\sqrt{\sigma},\,  A^{1/2}\sqrt{\rho}U\rangle_{HS}}\\
        &\leq \abs{\langle A^{1/2}\sqrt{\sigma},\, A^{1/2}\sqrt{\sigma}\rangle_{HS}}^{1/2} \times \abs{\langle A^{1/2}\sqrt{\rho}U,\,  A^{1/2}\sqrt{\rho}U\rangle_{HS}}^{1/2}\\
        &=\abs{\tr{A\sigma}}^{1/2}\times \abs{\tr{A\rho}}^{1/2}.
    \end{aligned}
\end{equation}
Dividing the lower bound in~\eqref{eq:gramian_inequality_mixed} by $\abs{\tr{A\rho}}^{1/2}$ leads to
\begin{equation}
    \abs{\tr{A\sigma}}^{1/2} \geq \sqrt{\cF(\rho,\,\sigma)}\sqrt{\langle A\rangle_\rho} - \frac{\Delta A_\rho}{\sqrt{\langle A\rangle_\rho}}\sqrt{1 - \cF(\rho,\,\sigma)}
\end{equation}
Under the condition that $\sqrt{\cF(\rho,\,\sigma) / (1 - \cF(\rho,\,\sigma))} \geq \Delta A_\rho / \langle A\rangle_\rho$, we can square both sides of the inequality
\begin{equation}
    \begin{aligned}
        \langle A\rangle_\sigma &\geq \left(\sqrt{\cF(\rho,\,\sigma)}\sqrt{\langle A\rangle_\rho} - \frac{\Delta A_\rho}{\sqrt{\langle A\rangle_\rho}}\sqrt{1 - \cF(\rho,\,\sigma)}\right)^2\\
        &= \cF(\rho,\,\sigma)\langle A\rangle_\rho - 2\Delta A_\rho \sqrt{\cF(\rho,\,\sigma)(1 - \cF(\rho,\,\sigma))} + \frac{1-\cF(\rho,\,\sigma)}{\langle A\rangle_\rho}(\Delta A_\rho)^2.
    \end{aligned}
\end{equation}
Notice that the RHS is monotonically decreasing as $\cF(\rho,\,\sigma)$ decreases. Hence, we can replace the true fidelity by a lower bound to it. In particular, for $\epsilon \geq 0$ with $\cF(\rho,\,\sigma) \geq 1 - \epsilon$ and $\sqrt{(1-\epsilon) / \epsilon} \geq \Delta A_\rho / \langle A\rangle_\rho$ we get
\begin{equation}
    \langle A\rangle_\sigma \geq (1 - 2\epsilon)\langle A\rangle_\rho -2 \sqrt{\epsilon(1-\epsilon)}\Delta A_\rho + \frac{\epsilon\langle A^2\rangle_\rho}{\langle A\rangle_\rho}
\end{equation}
which is the desired result.
\end{proof}

\section{Proof of Theorem~\ref{thm:gramian_eigenvalue}}
\label{sec:appendix-gramian-eigenvalue}
\begin{customthm}{\ref{thm:gramian_eigenvalue}}[restated]
    Let $\rho\in\cS(\cH_d)$ be a density operator and let $A$ be an arbitrary observable with eigenstate $\ket{\psi}$ and eigenvalue $\lambda$, $A\ket{\psi} = \lambda\ket{\psi}$. Suppose that $\epsilon \geq 0$ is such that $\cF(\rho,\,\ket{\psi}) = \bra{\psi}\rho\ket{\psi} \geq 1-\epsilon$. Then, lower and upper bounds for $\lambda$ can be expressed as
    \begin{equation}
        \langle A \rangle_\rho - \Delta A_\rho \sqrt{\frac{\epsilon}{1 - \epsilon}} \ \leq \ \lambda \ \leq \ \langle A \rangle_\rho + \Delta A_\rho \sqrt{\frac{\epsilon}{1 - \epsilon}}.
    \end{equation}
\end{customthm}
\begin{proof}
    Recall that in the proof of Theorem~\ref{thm:gramian_expectation} we have shown that a slight modification of the Gramian inequalities from~\eqref{eq:gramian_inequality} also holds for mixed states. Specifically, we have shown that
    \begin{equation}
    \label{eq:gramian_inequality_eigenvalue}
        \begin{aligned}
            \sqrt{\cF(\rho,\,\sigma)}\langle A\rangle_\rho - \Delta A_\rho\sqrt{1 - \cF(\rho,\,\sigma)} \ \leq \ \Re(\langle A\sqrt{\sigma},\,  \sqrt{\rho}U\rangle_{HS}) \ \leq \ \sqrt{\cF(\rho,\,\sigma)}\langle A\rangle_\rho + \Delta A_\rho\sqrt{1 - \cF(\rho,\,\sigma)}.
        \end{aligned}
\end{equation}
where $\langle\cdot,\,\cdot\rangle_{HS}$ denotes the Hilbert-Schmidt inner product and $U$ is a unitary such that $\cF(\rho,\,\sigma) = \abs{\braket{\Phi}{\Psi}}^2$ with
\begin{equation}
    \ket{\Psi} \equiv (\sqrt{\sigma} \tens \Id) \ket{\Omega},\hspace{2em} \ket{\Phi} \equiv (\sqrt{\rho} \tens U^T) \ket{\Omega}.
\end{equation}
Here, by assumption $\sigma=\ketbra{\psi}{\psi}$ is pure with $\ket{\psi}$ an eigenstate of $A$ with eigenvalue $\lambda$, $A\ket{\psi} = \lambda\ket{\psi}$. Note that in this case
\begin{equation}
    \begin{aligned}
        \langle A\sqrt{\sigma},\,  \sqrt{\rho}U\rangle_{HS} &= \lambda \ \langle \sqrt{\sigma},\, \sqrt{\rho}U\rangle_{HS}\\
        &=\lambda \ \tr{\sqrt{\sigma}\sqrt{\rho}U}\\
        &=\lambda \ \bra{\Omega}(\sqrt{\sigma}\sqrt{\rho}U \tens \Id)\ket{\Omega}\\
        &=\lambda \ \bra{\Omega}(\sqrt{\sigma}\sqrt{\rho} \tens U^T)\ket{\Omega}\\
        &=\lambda \ \braket{\Psi}{\Phi}
    \end{aligned}
\end{equation}
where $\ket{\Phi}$ and $\ket{\Psi}$ are the purifications of $\rho$ and $\sigma$ given in the proof of Theorem~\ref{thm:gramian_expectation} in~\eqref{eq:purifications}. Without loss of generality, we assume that $\braket{\Psi}{\Phi}$ is real and positive, since otherwise each state can be multiplied by a global phase. Dividing each side in~\eqref{eq:gramian_inequality_eigenvalue} by $\braket{\Psi}{\Phi}$ and noting that $\sqrt{\cF(\rho,\,\sigma)} = \braket{\Psi}{\Phi}$ yields
\begin{equation}
    \langle A\rangle_\rho - \Delta A_\rho\sqrt{\frac{1 - \cF(\rho,\,\sigma)}{\cF(\rho,\,\sigma)}} \ \leq \ \lambda \ \leq \ \langle A\rangle_\rho + \Delta A_\rho\sqrt{\frac{1 - \cF(\rho,\,\sigma)}{\cF(\rho,\,\sigma)}}.
\end{equation}
Since the RHS (LHS) of this inequality is monotonically increasing (decreasing) as $\cF(\rho,\,\sigma)$ decreases, we can replace the exact fidelity by an upper bound and still get valid bounds. That is, for $\epsilon \geq 0$ with $\cF(\rho,\,\sigma) \geq 1 - \epsilon$, we have
\begin{equation}
    \langle A \rangle_\rho - \Delta A_\rho \sqrt{\frac{\epsilon}{1 - \epsilon}} \ \leq \ \lambda \ \leq \ \langle A \rangle_\rho + \Delta A_\rho \sqrt{\frac{\epsilon}{1 - \epsilon}}.
\end{equation}
which is the desired result.
\end{proof}

\section{Fidelity Estimation}
\label{sec:appendix-fidelity}
\subsection{Proofs}
Here we give proofs for the fidelity lower bounds reported in Section~\ref{sec:fidelity_estimation}. In the sequel, let $H$ be a Hamiltonian with spectral decomposition
\begin{align}
    H = \sum_{i = 0}^m \lambda_i \Pi_i
\end{align}
where $\lambda_i$ are the eigenvalues (in increasing order), $\Pi_i$ is the projections onto the eigenspace associated with $\lambda_i$ and $m$ is the number of distinct eigenvalues. We write $\mathrm{Eig}_H(\lambda_i)$ for the space spanned by eigenvectors of $H$ with eigenvalue $\lambda_i$. In the following we first consider the non-degenerate case, that is when $\mathrm{Eig}_H(\lambda_0)$ is of dimension $1$ and treat the degenerate case separately.
\subsubsection{Non-degenerate case.}
We first consider the non-degenerate case, in which case $\Pi_0 = \ketbra{\psi_0}{\psi_0}$.\\

\paragraph{Eckart's Criterion.}
Eckart's criterion~\cite{eckart1930} is a method to lower bound the fidelity of an approximate state $\rho$ with one of the ground states of the Hamiltonian $H$. We include the proof here for completeness. For general $H$ and $\rho$, note that
\begin{align}
    \langle H - \lambda_0\Id_d\rangle_\rho &= \sum_{n = 1}^m(\lambda_i - \lambda_0)\tr{\Pi_i\rho}\\
    &\geq (\lambda_1 - \lambda_0)(1 - \bra{\psi_0}\rho\ket{\psi_0})
\end{align}
and thus
\begin{equation}
    \label{eq:eckart_criterion_proof}
    \bra{\psi_0}\rho\ket{\psi_0} \geq \frac{\lambda_1 - \langle H\rangle_\rho}{\lambda_1 - \lambda_0}.
\end{equation}\\

\paragraph{Bounds from \eqref{eq:mcclean_fidelity_bound} \& \eqref{eq:fidelity_bound}.}
The fidelity bound from eq.~\eqref{eq:mcclean_fidelity_bound} has been shown in~\cite{mcclean2016theory} for pure states. Here, we
extend this to mixed states and will discuss the degenerate case in the next section. Recall that $\delta$ is a lower bound on the spectral gap, $\lambda_1 - \lambda_0 \geq \delta$. Note that
\begin{align}
    \langle H\rangle_\rho &= \lambda_0 \bra{\psi_0}\rho\ket{\psi_0} + \sum_{i=1}^m \lambda_i \tr{\Pi_i\rho}\\
    &\geq \lambda_0 \bra{\psi_0}\rho\ket{\psi_0} + \sum_{i=1}^m (\lambda_0 + \delta)\tr{\Pi_i\rho}\\
    &=\lambda_0 \bra{\psi_0}\rho\ket{\psi_0} + (\lambda_0 + \delta)(1 - \bra{\psi_0}\rho\ket{\psi_0})\\
    &= \lambda_0 + \delta(1 - \bra{\psi_0}\rho\ket{\psi_0}).\label{eq:mcclean_bound_derivation}
\end{align}
Since by assumption $\lambda_0$ is non-degenerate and $\langle H\rangle_\rho \leq \frac{1}{2}(\lambda_0 + \lambda_1)$ it follows from Eckart's condition that $\bra{\psi_0}\rho\ket{\psi_0} \geq \frac{1}{2}$. By plugging this lower bound into the Gramian eigenvalue bound (Theorem~\ref{thm:gramian_eigenvalue}), we recover Weinstein's lower bound~\cite{Weinstein1934} for mixed states
\begin{equation}
    \lambda_0 \geq \langle H\rangle_\rho - \Delta H_\rho
\end{equation}
where $(\Delta H_\rho)^2$ is the variance of $H$. Using this to lower bound $\lambda_0$ in~\eqref{eq:mcclean_bound_derivation} and rearranging terms leads to the bound in~\eqref{eq:mcclean_fidelity_bound}
\begin{equation}
    \bra{\psi_0}\rho\ket{\psi_0} \geq 1 - \frac{\Delta H_\rho}{\delta}.
\end{equation}
If, on the other hand, we lower bound $\lambda_0$ in~\eqref{eq:mcclean_bound_derivation} by the Gramian eigenvalue lower bound~(Theorem~\ref{thm:gramian_eigenvalue}), we obtain the inequality
\begin{equation}
    \bra{\psi_0}\rho\ket{\psi_0} - 1 + \frac{\Delta H_\rho}{\delta}\sqrt{\frac{1}{\bra{\psi_0}\rho\ket{\psi_0}} - 1} \geq 0.
\end{equation}
The left hand side can be rewritten as a cubic polynomial in $\bra{\psi_0}\rho\ket{\psi_0}$. Under the assumption that $\langle H\rangle_\rho \leq \frac{1}{2}(\lambda_0 + \lambda_1)$ we again use Eckart's condition to find that $\bra{\psi_0}\rho\ket{\psi_0} \geq \frac{1}{2}$. It then follows that the inequality is satisfied if
\begin{equation}
    \bra{\psi_0}\rho\ket{\psi_0} \geq \frac{1}{2}\left(1 + \sqrt{1 - \left(\frac{\Delta H_\rho}{\delta/2}\right)^2}\,\right)
\end{equation}
which is the bound given in~\eqref{eq:fidelity_bound}.

\subsubsection{Degenerate case}
\label{sec:appendix-fidelity-degenerate}
If $\lambda_0$ is degenerate, then $\Pi_0 = \sum_{j=0}^{d_0}\ketbra{\psi_{0,j}}{\psi_{0,j}}$ where $d_0$ denotes the dimensionality of the eigenspace associated with $\lambda_0$. In the following, we first show that if $\rho$ is a pure state, then there exists an element $\ket{\psi}\in\mathrm{Eig}_H(\lambda_0)$ for which each of the fidelity bounds holds. If, on the other hand, $\rho$ is allowed to be mixed, we construct a simple counterexample for which the fidelity bounds are violated.\\

\paragraph{Pure states.}
Suppose that $\rho$ is a pure state $\rho = \ketbra{\phi}{\phi}$.
For Eckart's criterion, an analogous calculation leads to
\begin{equation}
    \sum_{j=0}^{d_0}\abs{\braket{\psi_{0,j}}{\phi}}^2 \geq \frac{\lambda_1 - \langle H\rangle_\rho}{\lambda_1 - \lambda_0}.
\end{equation}
Consider the state
\begin{equation}
    \ket{\psi} = \Gamma^{-1/2}\sum_i \braket{\psi_{0,i}}{\phi}\ket{\psi_{0,i}},\ws\ws \Gamma = \sum_i\abs{\braket{\psi_{0,i}}{\phi}}^2
\end{equation}
and note that $\braket{\psi}{\psi} = 1$ and $\ket{\psi}\in\mathrm{Eig}_H(\lambda_0)$. Furthermore, we have
\begin{equation}
    \abs{\braket{\psi}{\phi}}^2 = \sum_{j=0}^{d_0} \abs{\braket{\psi_{0,j}}{\phi}}^2
\end{equation}
and hence
\begin{equation}
    \abs{\braket{\psi}{\phi}}^2 \geq \frac{\lambda_1 - \langle H\rangle_\rho}{\lambda_1 - \lambda_0}.
\end{equation}
so that Eckart's criterion holds in the degenerate case for this particular choice of eigenstate $\ket{\psi}$ and for pure approximation states $\ket{\phi}$. Using again analogous calculations, we also obtain the extensions of the bounds \eqref{eq:mcclean_fidelity_bound} and~\eqref{eq:fidelity_bound} for pure approximation state $\ket{\phi}$ in the degenerate case and for the same choice of eigenstate $\ket{\psi}$.\\

\paragraph{Counterexample for mixed states.}
If the approximation state $\rho$ is allowed to be arbitrarily mixed, the above fidelity bounds do not hold in general. Indeed, consider the Hamiltonian
\begin{equation}
    H = U\cdot
    \begin{pmatrix}
    \lambda & 0 & 0\\
    0 & \lambda & 0\\
    0 & 0 & \mu\\
    \end{pmatrix}\cdot U^\dagger
\end{equation}
for arbitrary $\lambda,\,\mu\in\R$ with $\lambda < \mu$ and some arbitrary unitary $U$. Furthermore, let $\rho$ be the maximally mixed state $\rho = \frac{1}{3}\Id_3$ and note that $\langle H\rangle_\rho = \frac{2\lambda + \mu}{3}$. Thus, for any $\ket{\psi}$ we find that
\begin{equation}
    \bra{\psi}\rho\ket{\psi} = \frac{1}{3} < \frac{2}{3} = \frac{\mu - \langle H\rangle_\rho}{\mu - \lambda}
\end{equation}
 in violation of Eckart's criterion. To see that we can also construct a counterexample for the other two bounds, we calculate the variance
 \begin{align}
     (\Delta H_\rho)^2 &= \langle H^2\rangle_\rho - \langle H\rangle_\rho^2 = \frac{2(\mu - \lambda)^2}{9}
 \end{align}
 and notice that
 \begin{equation}
     \bra{\psi}\rho\ket{\psi} = \frac{1}{3} < 1 - \frac{\sqrt{2}}{3} = 1 - \frac{\Delta H_\rho}{\delta}
 \end{equation}
 and similarly
 \begin{equation}
     \bra{\psi}\rho\ket{\psi} = \frac{1}{3} < \frac{2}{3} = \frac{1}{2}\left(1 + \sqrt{1 - \left(\frac{\Delta H_\rho}{\delta/2}\right)^2}\,\right).
 \end{equation}

\newpage
\section{Simulations with higher error rates}
\label{sec:higher-error-rates}

\begin{table}[ht]
    \centering
    \begin{tabularx}{\textwidth}{@{}c c c c X c c X c X c c @{}}
        \toprule
        \multirow{2}{*}{\makecell{Bond\\Distance (\AA)}} & \multirow{2}{*}{$E_0$} & \multirow{2}{*}{\makecell{VQE\\(SPA)}} & \multirow{2}{*}{Fidelity} && \multicolumn{2}{c}{Gramian Eigenvalue} && \multicolumn{1}{c}{Gramian Expectation} && \multicolumn{2}{c}{SDP} \\
        \cmidrule(lr){6-7}\cmidrule(lr){9-9}\cmidrule(lr){11-12}
        & & & && lower bound & upper bound && lower bound && lower bound & upper bound\\
        \midrule
        $0.50$ & $-7.21863$ & $-6.88805$ & $0.642$ && $-7.25193$ & $-6.52475$ && $-7.32715$ && $-7.34604$ & $-5.78115$\\
        $0.75$ & $-7.70845$ & $-7.29597$ & $0.643$ && $-7.74638$ & $-6.84531$ && $-7.83840$ && $-7.86084$ & $-5.93280$\\
        $1.00$ & $-7.90403$ & $-7.49060$ & $0.645$ && $-7.94329$ & $-7.03703$ && $-8.03507$ && $-8.05745$ & $-6.10197$\\
        $1.25$ & $-7.97808$ & $-7.57789$ & $0.642$ && $-8.01822$ & $-7.13739$ && $-8.10760$ && $-8.12808$ & $-6.23202$\\
        $1.40$ & $-7.99541$ & $-7.61078$ & $0.645$ && $-8.03193$ & $-7.18843$ && $-8.12397$ && $-8.14492$ & $-6.30208$\\
        $1.50$ & $-8.00062$ & $-7.62723$ & $0.645$ && $-8.03743$ & $-7.21626$ && $-8.12843$ && $-8.14912$ & $-6.34578$\\
        $1.60$ & $-8.00251$ & $-7.63936$ & $0.646$ && $-8.03613$ & $-7.24368$ && $-8.12772$ && $-8.14878$ & $-6.40143$\\
        $1.70$ & $-8.00213$ & $-7.65290$ & $0.644$ && $-8.03798$ & $-7.26681$ && $-8.12450$ && $-8.14473$ & $-6.45257$\\
        $1.75$ & $-8.00127$ & $-7.65684$ & $0.646$ && $-8.03181$ & $-7.28154$ && $-8.12101$ && $-8.14171$ & $-6.49063$\\
        $2.00$ & $-7.99319$ & $-7.68221$ & $0.644$ && $-8.02412$ & $-7.34029$ && $-8.10203$ && $-8.12208$ & $-6.63209$\\
        $2.25$ & $-7.98161$ & $-7.69596$ & $0.639$ && $-8.01110$ & $-7.38116$ && $-8.08242$ && $-8.10091$ & $-6.75887$\\
        $2.50$ & $-7.96941$ & $-7.70682$ & $0.634$ && $-7.99792$ & $-7.41511$ && $-8.06577$ && $-8.08292$ & $-6.85437$\\
        $2.75$ & $-7.95765$ & $-7.71482$ & $0.632$ && $-7.98661$ & $-7.44292$ && $-8.04982$ && $-8.06660$ & $-6.94268$\\
        $3.00$ & $-7.94669$ & $-7.71586$ & $0.624$ && $-7.97520$ & $-7.45615$ && $-8.06343$ && $-8.07915$ & $-6.96150$\\
        $3.25$ & $-7.93706$ & $-7.72211$ & $0.623$ && $-7.96702$ & $-7.47666$ && $-8.09013$ && $-8.11182$ & $-6.95623$\\
        $3.50$ & $-7.92842$ & $-7.72265$ & $0.613$ && $-7.96128$ & $-7.48390$ && $-8.10708$ && $-8.13833$ & $-6.94904$\\
        $3.75$ & $-7.92096$ & $-7.72526$ & $0.606$ && $-7.95586$ & $-7.49476$ && $-8.11762$ && $-8.16116$ & $-6.94122$\\
        $4.00$ & $-7.91490$ & $-7.72740$ & $0.597$ && $-7.95494$ & $-7.50039$ && $-8.13010$ && $-8.18053$ & $-6.93847$\\
        $4.25$ & $-7.90968$ & $-7.72717$ & $0.583$ && $-7.95372$ & $-7.50055$ && $-8.14249$ && $-8.19643$ & $-6.93280$\\
        $4.50$ & $-7.90590$ & $-7.72953$ & $0.568$ && $-7.95324$ & $-7.50594$ && $-8.15344$ && $-8.21024$ & $-6.92201$\\
        $4.75$ & $-7.90306$ & $-7.73045$ & $0.551$ && $-7.95891$ & $-7.50178$ && $-8.16757$ && $-8.22021$ & $-6.92056$\\
        $5.00$ & $-7.90106$ & $-7.72925$ & $0.527$ && $-7.96652$ & $-7.49266$ && $-8.18161$ && $-8.22696$ & $-6.91288$\\
        $5.25$ & $-7.89982$ & $-7.73207$ & $0.516$ && $-7.96729$ & $-7.49630$ && $-8.18806$ && $-8.23073$ & $-6.91086$\\
     \bottomrule
    \end{tabularx}
    \caption{Noisy simulations of VQE for ground state energies of LiH(2, 4) with a separable pair approximation ansatz (SPA). The noise model consists of bitflip on single qubit gates and depolarization error on two qubit gates. The error probability for both noise channels is set to $10\%$.}
    \label{tab:lih_spa_increased_noise}
\end{table}

\begin{table}[ht]
    \centering
    \begin{tabularx}{\textwidth}{@{}c c c c X c c X c X c c @{}}
        \toprule
        \multirow{2}{*}{\makecell{Bond\\Distance (\AA)}} & \multirow{2}{*}{$E_0$} & \multirow{2}{*}{\makecell{VQE\\(UpCCGSD)}} & \multirow{2}{*}{Fidelity} && \multicolumn{2}{c}{Gramian Eigenvalue} && \multicolumn{1}{c}{Gramian Expectation} && \multicolumn{2}{c}{SDP} \\
        \cmidrule(lr){6-7}\cmidrule(lr){9-9}\cmidrule(lr){11-12}
        & & & && lower bound & upper bound && lower bound && lower bound & upper bound\\
        \midrule
        $0.50$ & $-7.21863$ & $-6.66418$ & $0.118$ && $-7.71356$ & $-5.61455$ && $-7.34604$ && $-7.34604$ & $-5.48348$\\
        $0.75$ & $-7.70845$ & $-6.99642$ & $0.120$ && $-8.30129$ & $-5.69055$ && $-7.86084$ && $-7.86084$ & $-5.55101$\\
        $1.00$ & $-7.90403$ & $-7.17549$ & $0.122$ && $-8.49462$ & $-5.85839$ && $-8.05745$ && $-8.05745$ & $-5.71484$\\
        $1.25$ & $-7.97808$ & $-7.27489$ & $0.121$ && $-8.56334$ & $-5.98558$ && $-8.12808$ && $-8.12808$ & $-5.85695$\\
        $1.40$ & $-7.99541$ & $-7.30592$ & $0.107$ && $-8.61254$ & $-6.00004$ && $-8.14492$ && $-8.14492$ & $-5.93376$\\
        $1.50$ & $-8.00062$ & $-7.34297$ & $0.117$ && $-8.56754$ & $-6.12004$ && $-8.14912$ && $-8.14912$ & $-5.99011$\\
        $1.60$ & $-8.00251$ & $-7.37233$ & $0.120$ && $-8.53701$ & $-6.20478$ && $-8.14878$ && $-8.14878$ & $-6.05390$\\
        $1.70$ & $-8.00213$ & $-7.39912$ & $0.120$ && $-8.52712$ & $-6.26860$ && $-8.14473$ && $-8.14473$ & $-6.12434$\\
        $1.75$ & $-8.00127$ & $-7.41376$ & $0.121$ && $-8.51673$ & $-6.31010$ && $-8.14171$ && $-8.14171$ & $-6.16141$\\
        $2.00$ & $-7.99319$ & $-7.47053$ & $0.120$ && $-8.46034$ & $-6.48104$ && $-8.12208$ && $-8.12208$ & $-6.34924$\\
        $2.25$ & $-7.98161$ & $-7.51937$ & $0.120$ && $-8.41918$ & $-6.62024$ && $-8.10091$ && $-8.10091$ & $-6.51810$\\
        $2.50$ & $-7.96941$ & $-7.55501$ & $0.118$ && $-8.38562$ & $-6.72552$ && $-8.08292$ && $-8.08292$ & $-6.65920$\\
        $2.75$ & $-7.95765$ & $-7.58323$ & $0.116$ && $-8.36295$ & $-6.80337$ && $-8.06660$ && $-8.06660$ & $-6.77282$\\
        $3.00$ & $-7.94669$ & $-7.60655$ & $0.116$ && $-8.33990$ & $-6.87168$ && $-8.07915$ && $-8.07915$ & $-6.81481$\\
        $3.25$ & $-7.93706$ & $-7.62472$ & $0.115$ && $-8.32053$ & $-6.92821$ && $-8.11182$ && $-8.11182$ & $-6.82651$\\
        $3.50$ & $-7.92842$ & $-7.64013$ & $0.115$ && $-8.31261$ & $-6.96878$ && $-8.13833$ && $-8.13833$ & $-6.83591$\\
        $3.75$ & $-7.92096$ & $-7.65034$ & $0.103$ && $-8.33392$ & $-6.96726$ && $-8.16116$ && $-8.16116$ & $-6.84353$\\
        $4.00$ & $-7.91490$ & $-7.66562$ & $0.111$ && $-8.31035$ & $-7.02064$ && $-8.18053$ && $-8.18053$ & $-6.85062$\\
        $4.25$ & $-7.90968$ & $-7.68425$ & $0.089$ && $-8.34869$ & $-7.01927$ && $-8.19643$ && $-8.19643$ & $-6.85685$\\
        $4.50$ & $-7.90590$ & $-7.69349$ & $0.089$ && $-8.35523$ & $-7.03175$ && $-8.21024$ && $-8.21024$ & $-6.86264$\\
        $4.75$ & $-7.90306$ & $-7.70511$ & $0.090$ && $-8.35773$ & $-7.05194$ && $-8.22021$ && $-8.22021$ & $-6.86835$\\
        $5.00$ & $-7.90106$ & $-7.70031$ & $0.098$ && $-8.36531$ & $-7.03541$ && $-8.22696$ && $-8.22696$ & $-6.87387$\\
        $5.25$ & $-7.89982$ & $-7.70455$ & $0.077$ && $-8.45740$ & $-6.95222$ && $-8.23073$ && $-8.23073$ & $-6.87932$\\
        \bottomrule
    \end{tabularx}
    \caption{Noisy simulations of VQE for ground state energies of LiH(2, 4) with an UpCCGSD ansatz. The noise model consists of bitflip on single qubit gates and depolarization error on two qubit gates. The error probability for both noise channels is set to $10\%$.}
    \label{tab:lih_upccgsd_increased_noise}
\end{table}

\newpage
\section{Simulations with realistic noise models}
\label{sec:realistic-noise-models}
\begin{figure}[h]
    \centering
    \includegraphics[width=0.5\linewidth]{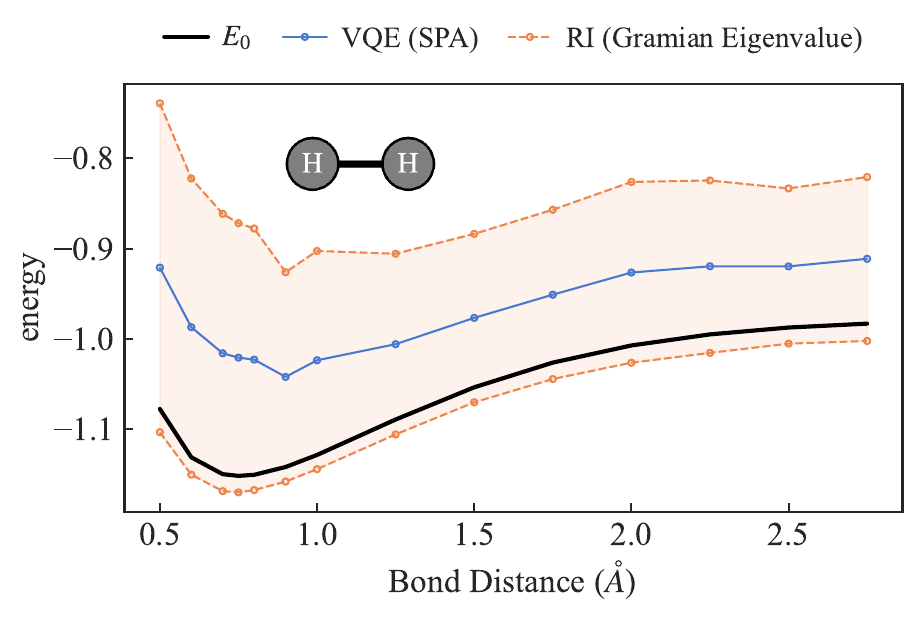}
    \caption{Bond dissociation curves and robustness interval (RI) for H$_2$ in a basis-set-free approach~\cite{kottmann2021reducing,kottmann2020direct}, using the noise model of the 5-Qubit \texttt{ibmq\_vigo} processor, one of the IBM Quantum Canary processors~\cite{ibmq}.}
    \label{fig:h2_vigo}
\end{figure}

\end{document}